\documentclass[american,aps,pra,reprint,floatfix,nofootinbib,superscriptaddress]{revtex4-2}
\usepackage[unicode=true,pdfusetitle, bookmarks=true,bookmarksnumbered=false,bookmarksopen=false, breaklinks=false,pdfborder={0 0 0},backref=false,colorlinks=false]{hyperref}
\hypersetup{colorlinks,linkcolor=myurlcolor,citecolor=myurlcolor,urlcolor=myurlcolor}
\usepackage{graphics,epstopdf,graphicx,amsthm,amsmath,amssymb,mathptmx,braket,colortbl,color,bm,framed,mathrsfs}
\usepackage{tikz}
\usepackage{qcircuit}
\usepackage{orcidlink}
\usepackage[T1]{fontenc}
\usepackage[up]{subfigure}

\definecolor{myurlcolor}{rgb}{0,0,0.9}

\newcommand{\proj}[1]{| #1\rangle\!\langle #1 |}

\newcommand{\inner}[2]{\langle #1 , #2\rangle}

\DeclareMathOperator{\trace}{Tr}
\newcommand{\Ptr}[2]{\trace_{#1}\Pa{#2}}
\newcommand{\Tr}[1]{\Ptr{}{#1}}

\newcommand{\Pa}[1]{\left[#1\right]}

\newcommand{\norm}[1]{\left\lVert #1 \right\rVert}

\theoremstyle{plain}
\newtheorem{thm}{Theorem}
\newtheorem{lem}[thm]{Lemma}
\newtheorem{prop}[thm]{Proposition}
\newtheorem{cor}[thm]{Corollary}

\theoremstyle{definition}
\newtheorem{example}{Example}

\newcommand*{\myproofname}{Proof}

\def\ot{\otimes}

\def\real{\mathbb{R}}

\def \supp {\mathrm{supp}}

\DeclareMathAlphabet{\mathcal}{OMS}{cmsy}{m}{n}

\makeatother

\begin{document}

 \author{Kaifeng Bu}
\email{kfbu@fas.harvard.edu}
\affiliation{Department of Physics, Harvard University, Cambridge, Massachusetts 02138, USA}

\author{Dax Enshan Koh\orcidlink{0000-0002-8968-591X}}
 \email{dax\textunderscore koh@ihpc.a-star.edu.sg}
\affiliation{Institute of High Performance Computing, Agency for Science, Technology and Research (A*STAR), 1 Fusionopolis Way, \#16-16 Connexis, Singapore 138632, Singapore}

\author{Roy J. Garcia}
\email{roygarcia@g.harvard.edu}
\affiliation{Department of Physics, Harvard University, Cambridge, Massachusetts 02138, USA}

\author{Arthur Jaffe}
\email{arthur\textunderscore jaffe@harvard.edu}
\affiliation{Department of Physics, Harvard University, Cambridge, Massachusetts 02138, USA}

\title{Classical shadows with Pauli-invariant unitary ensembles}

\begin{abstract}
The classical shadow estimation protocol is a noise-resilient and sample-efficient quantum algorithm for learning the properties of quantum systems. Its performance depends on the choice of a unitary ensemble, which must be chosen by a user in advance. What is the weakest assumption that can be made on the chosen unitary ensemble that would still yield meaningful and interesting results? To address this question, we consider the class of  Pauli-invariant unitary ensembles, i.e.~unitary ensembles that are invariant under multiplication by a Pauli operator. This class includes many previously studied ensembles like the local and global Clifford ensembles as well as locally scrambled unitary ensembles. For this class of ensembles, we provide an explicit formula for the reconstruction map corresponding to the shadow channel and give explicit sample complexity bounds. 
In addition, we provide two applications of our results. Our first application is to locally scrambled unitary ensembles, where we give explicit formulas for the reconstruction map and sample complexity bounds that circumvent the need to solve an exponential-sized linear system. Our second application is to the classical shadow tomography of quantum channels with Pauli-invariant unitary ensembles. Our results pave the way for more efficient or robust protocols for predicting important properties of quantum states, such as their fidelity, entanglement entropy, and quantum Fisher information.

\end{abstract}

\maketitle

\section{Introduction}
Learning the properties of an unknown but physically accessible quantum system is a fundamental task in quantum information processing \cite{huang2020predicting}. A standard tool for this task is quantum tomography, a process by which one recovers a classical description of a quantum system through  performing measurements on it. Unfortunately, finding the full description of a quantum system by quantum tomography is  computationally-intensive; it  requires an exponential number of copies of the system \cite{aaronson2019shadow, haah2017sample,o2016efficient}.

Recently, Huang, Kueng, and Preskill introduced a novel method---the \textit{classical shadow} paradigm~\cite{huang2020predicting}---to circumvent the above limitation. A key insight behind classical shadows rests on the fact that in many cases one does not need to learn a complete description of a quantum system; one can learn its most useful properties from
a minimal sketch  of the quantum system, the \textit{classical shadow}.

The performance of the classical shadow protocol depends on several factors. In particular, it depends on an ensemble of unitary operators from which an operation is chosen randomly to be applied to the unknown quantum state. 
A user must choose this ensemble in advance, according to some desiderata, such as the need for the shadow channel to be invertible and for one to have an efficient algorithm for sampling a unitary from the ensemble. Several unitary ensembles have been considered by previous authors, including the local and global Clifford ensembles~\cite{huang2020predicting}, fermionic Gaussian unitaries~\cite{zhao2021fermionic}, chaotic Hamiltonian evolutions~\cite{hu2021hamiltonian}, and locally scrambled unitary ensembles~\cite{hu2021classical}.

One can ask: what is the weakest assumption on the unitary ensemble that would still yield meaningful and interesting results? Our candidate solution to this question is the assumption that the unitary ensemble is invariant under multiplication by a Pauli operator. Ensembles satisfying this assumption---namely the Pauli-invariant unitary ensembles---include a wide range of ensembles, including the Pauli group and the aforementioned local and global Clifford group and locally scrambled unitary ensembles.

Let us summarize our main contributions in this work: For Pauli-invariant unitary ensembles, we find  an explicit formula for the reconstruction map corresponding to the shadow channel. We establish a connection between the coefficients of the reconstruction map and the entanglement features of the dynamics. We  give upper bounds on the sample complexity in terms of the shadow norm for the task of expectation estimation using the classical shadow protocol.

We also give two applications of our methods. First, we apply our results to locally scrambled unitary ensembles, giving the reconstruction map explicitly, along with sample complexity bounds. Unlike the approach taken in \cite{hu2021classical}, our study circumvents the need to solve an exponential-sized system of linear equations. Second, we apply our results to the shadow process tomography of quantum channels, where we generalize the results of \cite{levy2021classical, kunjummen2021shadow} to allow for any Pauli-invariant unitary ensemble. 

\noindent\textbf{\textit{Related work}}---The classical shadow paradigm, with the sample efficiency it touts, has attracted considerable attention over the last couple of years \cite{huang2022learning}. Several applications have been proposed, including applying the classical shadow framework to estimate expectation values of molecular Hamiltonians~\cite{hadfield2020measurements, hadfield2021adaptive}, detecting or estimating the degree of entanglement in quantum systems~\cite{huang2020predicting,elban2020mixed, neven2021symmetry, rath2021quantum}, classifying quantum data~\cite{li2021vsql}, measuring out-of-time-ordered correlators~\cite{garcia2021quantum},  approximating wave function overlaps~\cite{huggins2021unbiasing}, solving quantum many-body problems~\cite{huang2020predicting, huang2021provably, notarnicola2021randomized}, estimating gate-set properties~\cite{helsen2021estimating}, and avoiding barren plateaus~\cite{sack2022avoiding}. Noise analyses of the classical shadow protocol have been performed, with noise-resilient versions developed~\cite{chen2021robust,koh2020classical,flammia2021averaged}. 

Authors have proposed variants and generalizations of the protocol, including locally-biased classical shadows~\cite{hadfield2020measurements}, derandomization techniques~\cite{huang2021efficient}, decision diagram techniques~\cite{hillmich2021decision}, grouping strategies~\cite{wu2021overlapped},
POVM-based classical shadows~\cite{acharya2021informationally}, classical shadows with locally-scrambled unitary ensembles~\cite{hu2021classical}, and classical shadows with unitary ensembles generated by a Hamiltonian~\cite{hu2021hamiltonian}. Extensions have been found to fermions~\cite{zhao2021fermionic} and to quantum channels~\cite{levy2021classical,kunjummen2021shadow}. With respect to certain figures-of-merit, the protocol has been shown to give an advantage over alternative approaches \cite{hadfield2020measurements,lukens2021classical} and has been analyzed theoretically from a Bayesian point of view \cite{lukens2021bayesian}.
Lower bounds on the performance of classical shadows have also been proved~\cite{huang2020predicting,chen2021exponential}.

\section{Background: Classical shadows}

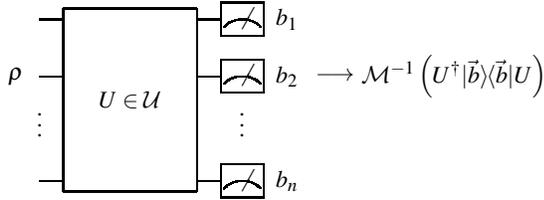
\begin{figure}
    \centering
    \begin{align*}
\Qcircuit @C=1em @R=1em {
& &    & \multigate{4}{ \quad   U\in\mathcal U  \quad }  &   \meter & b_1 \\ 
&\rho  &       &\ghost{ \quad U\in\mathcal U \quad } &  \meter & b_2 & & \longrightarrow  & &  & & &  \mathcal M^{-1}\left(U^\dag \proj{\vec{b}}U\right)
& & & & & \\
&  &  \vdots & & \vdots \\
& & &  \\
&  &     &\ghost{ \quad U\in\mathcal U \quad }  & \meter  & b_n \\    
}  
\end{align*}
    \caption{Circuit diagram representing a single round of the classical shadow protocol. The output of the protocol is the classical shadow given by $\mathcal M^{-1}\left(U^\dag \proj{\vec{b}}U\right)$.}
    \label{fig:classical_shadows}
\end{figure}

Let us begin by describing the classical shadow protocol (see also Figure 1, where we present a circuit diagram for a single round of the protocol) that was introduced in \cite{huang2020predicting}. The input of the protocol is an unknown but physically accessible quantum state $\rho$ and the output of the protocol is a set of classical shadows $\{\hat{\rho}_i\}_i$ that are unbiased estimators of $\rho$, i.e.~the expectation value of each $\hat{\rho}_i$ is precisely $\rho$. In an estimation task, these classical shadows may subsequently be used to estimate expectation values of given quantum observables in the state $\rho$.

Before the protocol is run, a user will have to pre-decide on a unitary ensemble $\mathcal{E}=\set{(U,P(U))}_{U\in \mathcal U}$, where $\mathcal U$ is a set of unitary operators and $P$ is a probability distribution on $\mathcal U$. The first step of the protocol involves choosing a random unitary $U$ from the unitary ensemble $\mathcal{E}$ according to the specified distribution $P(U)$. The unitary $U$ is then applied to $\rho$ and a computational basis measurement is performed on the resultant state to obtain an $n$-bit string $\vec b = b_1b_2\ldots b_n \in \{0,1\}^n$. The random unitary $U$ and random string $\vec b$ obtained are then combined to give the state \begin{align}
\hat{\sigma}_{U,\vec b}
=U^\dag \proj{\vec{b}}U,
\end{align}
which is stored in classical memory. Note that this produces an ensemble of states
\begin{align}
\mathcal{E}_{\rho}
=\left\{
\left(\hat{\sigma}_{U,\vec b},
P(U, \vec b)
\right)
\right\}_{U,\vec b},
\label{eq:ensembleState}
\end{align}
where $P(U,\vec b)=P(U)P(\vec b|U)$ and $P(\vec b|U)=\Tr{\hat{\sigma}_{U,\vec b}\rho}$, as given by the Born rule. The expected value of this ensemble is given by
\begin{align}
    \sigma=\mathbb{E}
_{\hat{\sigma}\in \mathcal{E}_{\rho}}
\hat{\sigma}
=\mathbb{E}_U
\sum_{\vec b}
\hat{\sigma}_{U,\vec b}
\Tr{\hat{\sigma}_{U,\vec b}\rho}
:=\mathcal{M}[\rho],
\end{align}
where $\mathcal M$---called the \textit{shadow channel}---is a completely positive and trace-preserving map.
To construct the classical shadow, the shadow channel needs to be invertible. If the conditions for invertibility are met, the inverse of the shadow channel $\mathcal M^{-1}$, called the \textit{reconstruction map}, is applied to the classically stored $\hat{\sigma}_{U,\vec b}$
to obtain the classical snapshot $\hat\rho = \mathcal{M}^{-1}[\hat{\sigma}_{U,\vec b}]$, which is called the \textit{classical shadow}.
As required, $\hat{\rho}$ is an unbiased estimator of $\rho$:
\begin{align}
\rho=\mathcal{M}^{-1}[\sigma]
=\mathop{\mathbb{E}}_{\hat{\sigma}\in\mathcal{E}_{\rho}}
\mathcal{M}^{-1}[\hat{\sigma}] = \mathop{\mathbb E}\limits_{U,b} \mathcal M^{-1}(\hat{\sigma}_{U,b})
=\mathop{\mathbb E}\limits_{U,b} \left[\hat \rho\right]
.    
\end{align}

Note that while $\mathcal M^{-1}$ is trace-preserving, it is not necessarily completely positive (this is fine since $\hat\rho$ does not need to be a quantum state). One repeats this process multiple times to construct an ensemble $\{\hat{\rho}_i\}_i$ of classical shadows. The number of times this process needs to be repeated---which translates to the number of samples of $\rho$ needed---depends on the number $N$ of different properties of $\rho$ that one wants to estimate.  Since quantum measurements are destructive, one might think naively that one would need at least $N$ samples, where at least one sample is used for each property. Surprisingly, with classical shadows, an order of $\log N$ samples is sufficient \cite{huang2020predicting}.

As illustrated above, the reconstruction map  $\mathcal{M}^{-1}$ plays a pivotal role in the classical shadow protocol. However, for arbitrary unitary ensembles, no general closed-form analytic formula is known for $\mathcal{M}^{-1}$. Instead, prior to this study, it was only for a handful of unitary ensembles (e.g.~the local and global Clifford ensembles) that analytic expressions have been derived. In this study, We address this gap by presenting an explicit formula for the reconstruction map for the wide-class of  Pauli-invariant unitary ensembles, thereby enlarging the class of ensembles for which such expressions are known.



\section{Main results}

\subsection{Noiseless classical shadows}

Let us now introduce the Pauli-invariant unitary ensembles, and derive results for the performance of classical shadow protocols that utilize such ensembles. For simplicity, we shall start by assuming that the classical shadow protocol is implemented perfectly without any noise. This assumption will subsequently be relaxed in Section \ref{sec:noisy_classical_shadow}.


\textit{\textbf{Pauli-invariant unitary ensemble}}---We denote the set of Pauli operators on $n$ qubits by $\mathcal{P}_n=\set{P_{\vec{a}}=\ot_iP_{a_i} : \vec{a}\in V^n}$, where $V:= \mathbb Z_2 \times \mathbb Z_2 = \set{(0,0),(0,1),(1,0),(1,1)}$, and
$P_{(x,z)} = i^{xz} X^x Z^z$, i.e.~$P_{(0,0)}= \mathbb{I}$, $P_{(0,1)}=Z$, $P_{(1,0)}=X$, $P_{(1,1)}=Y$. Given a unitary ensemble 
$\mathcal{E}=\set{(U,P(U))}_U$, $\mathcal{E}$ is called \textit{Pauli-invariant}  if 
the probability distribution $P(U)$ satisfies the following (right) Pauli-invariant condition, that is 
\begin{align}
P(U)=P(UP_{\vec{\sigma}}), \quad\forall P_{\vec{\sigma}}\in \mathcal{P}_n.
\end{align}
This is a weak assumption, and many unitary ensembles satisfy this condition, for example, (1) the Pauli group $\set{\pm 1, \pm i} \times \mathcal P_n$;
(2) $D$-dimensional local, random quantum circuits \cite{harrow2018approximate};
(3) local Clifford unitaries; and
(4) global Clifford unitaries. 

For Pauli-invariant unitary ensembles, we will now provide an explicit formula for the shadow channel  $\mathcal{M}$ and the corresponding reconstruction map $\mathcal{M}^{-1}$, if it exists. For clarity of presentation, the detailed proofs of all the results in this subsection are provided in Appendix \ref{sec:noiseless_classical_shadows}.

\begin{thm}\label{thm:main1}
For a Pauli-invariant unitary ensemble $\mathcal{E}$, the shadow channel can be written as
\begin{align}
\mathcal{M}[\rho]
=\frac{1}{2^{n}}\sum_{\vec{a}\in V^n}
W_{\mathcal{E}}[\vec{a}]\Tr{\rho P_{\vec{a}}}P_{\vec{a}},
\end{align}
where 
$W_{\mathcal{E}}[\vec{a}]$ is the average squared Pauli coefficient of the classical shadow, defined as 

\begin{align}
W_{\mathcal{E}}[\vec{a}]=\mathbb{E}_{\vec b}\mathbb{E}_U
W_{\hat{\sigma}_{U,\vec b}}[\vec{a}],
\end{align}
where $\mathbb{E}_{\vec b}:=\frac{1}{2^n}\sum_{\vec b}$ denotes the average with respect to the uniform distribution over $n$-bit strings. For a state $\sigma$, we have used the notation  $W_\sigma[\vec{a}]:=|\Tr{\sigma P_{\vec{a}}}|^2$
to denote the square of the $\vec a$-th Pauli coefficient of $\sigma$.

For Pauli-invariant unitary ensembles, the reconstruction map $\mathcal{M}^{-1}$ exists if and only if
$W_{\mathcal{E}}[\vec{a}]>0$ for all $\vec{a}$. If it exists, the reconstruction map 
is completely specified by its action on the Pauli basis elements as follows:
\begin{align}
\mathcal{M}^{-1}[P_{\vec{a}}]
=W_{\mathcal{E}}[\vec{a}]^{-1}P_{\vec{a}}.
\end{align}

\end{thm}

Now, let us discuss the connection between the coefficients of the reconstruction map
$W_{\mathcal{E}}$ and the 2nd entanglement feature  \cite{you2018machine,you2018entanglement}, which has been used to describe the entangling power of unitary ensembles and is defined as 
follows:
\begin{align}
E^{(2)}_{\mathcal{E}}[A]
=\mathbb{E}_{U}\mathbb{E}_{\vec b}
e^{-S^{(2)}_A(\hat{\sigma}_{U,b})},
\end{align}
where $A$ is a subset of $[n]$,  $S^{(2)}_A(\sigma)=-\log\Tr{\sigma^2_A}$, where $\sigma_A=\text{Tr}_{A^c}{[\sigma]}$ denotes the 2nd R\'enyi entanglement entropy 
of the state $\sigma$  on the subset $A$. 
For any subset $S\subset [n]$, let us define $W_{\mathcal{E}}[S]$ as the sum of the 
coefficients whose support is $S$, i.e., 
\begin{align}
W_{\mathcal{E}}[S]=\sum_{\vec{a}:\supp(\vec{a})=S}
W_{\mathcal{E}}[\vec{a}],
\end{align}
where $\supp(\vec{a})$ denotes the support of the vector $\vec{a}$. The following proposition expresses $W_{\mathcal{E}}
[S]$ in terms of the 2nd entanglement feature.
\begin{prop}\label{prop:entfeature}
The coefficients of the reconstruction map can be expressed as follows:

\begin{align}
W_{\mathcal{E}}
[S]=(-1)^{|S|}
\sum_{A\subset [S]}
(-2)^{|A|}
E^{(2)}_{\mathcal{E}}[A].
\end{align}

\end{prop}

Now, let us consider the sample complexity for the task of expectation estimation using
classical shadows with Pauli-invariant unitary ensembles. We shall make use of the result from \cite{huang2020predicting} that the sample complexity of classical shadows with the input state $\rho$ and an observable $O$ is upper-bounded 
by the following (squared) shadow norm:
\begin{align}\label{ShadowNormDefn}
 \norm{O}^2_{\mathcal{E}_{\rho}}
=\mathbb{E}_{\mathcal{E}_{\rho}}\hat{o}[\hat{\sigma}]^2,
\end{align}
where the estimator $\hat{o}$ of the observable $O$ is given by $
\hat{o}
  =\Tr{O\mathcal{M}^{-1}[\hat{\sigma}]}
$.
More precisely, the sample complexity $\mathscr S$ needed to accurately predict a collection of $N$ linear 
target functions $\set{\Tr{O_i \rho}}^N_{i=1}$  with error $\epsilon$ and failure probability $\delta$ is \cite{huang2020predicting}
\begin{align*}
    \mathscr S = O\left(
    \frac{\log(N/\delta)}{\epsilon^2}
    \max_{1\leq i\leq N}
    \left\|
    O_i - \frac 1{2^n} \Tr{O_i} \mathbb{I}
    \right\|^2_{\mathcal E_\rho}
    \right).
\end{align*}

In general, the shadow norm is hard to compute, though for special cases, closed-form expressions for the shadow norm can be derived. To this end, let us first consider the special but widely-used example wherein the observable 
$O$ is taken to be a Pauli operator $P_{\vec{a}}$.

\begin{prop}
If the observable $O$ is the Pauli operator $P_{\vec{a}}$, then the shadow norm is equal to
\begin{align}
\norm{P_{\vec{a}}}^2_{\mathcal{E}_{\rho}}
=W_{\mathcal{E}}[\vec{a}]^{-1}.
\end{align}
\end{prop}

Note that the shadow norm $\norm{O}_{\mathcal{E}}$ usually depends on the input state $\rho$. 
If we interested in the expectation of the shadow norm over some unitary ensemble instead of 
some specific state, we will need a notion of an average shadow norm. To this end, let us consider the (squared) average shadow norm over the Pauli group, defined as
$
\norm{O}^2_{\mathcal{E}}
=\mathbb{E}_{V\in \mathcal{P}_n}
\norm{O}^2_{\mathcal{E}_{V\rho V^\dag}}.
$
The following proposition gives an expression for the average shadow norm.

\begin{prop}
The (squared) average shadow norm over the Pauli group
$\norm{O}^2_{\mathcal{E}}$ can be expressed as follows:
\begin{align}
\norm{O}^2_{\mathcal{E}}
=\frac{1}{4^n}
\sum_{\vec{a}\in V^n}W_{\mathcal{E}}[\vec{a}]^{-1}
W_O[\vec{a}],
\end{align}
where  $W_O[\vec{a}]=|\Tr{OP_{\vec{a}}}|^2$.
\end{prop}
Therefore, given a set of $N$ traceless observables $\set{O_i}^N_{i=1}$, the average shadow norm provides the following lower
bound for the sample complexity: 
\begin{align}
    \mathscr S \geq
    \frac{\log(N/\delta)}{\epsilon^2}
    \max_{1\leq i\leq N}
\inner{\vec{W}^{-1}_{\mathcal{E}}}{\vec{W}_{O_i}},
\end{align}
where the (normalized) 
inner product is defined as  $\inner{\vec{W}^{-1}_{\mathcal{E}}}{\vec{W}_{O_i}}:=\frac{1}{4^n}
\sum_{\vec{a}\in V^n}W_{\mathcal{E}}[\vec{a}]^{-1}
W_{O_i}[\vec{a}]$.

\subsection{Noisy classical shadows}
\label{sec:noisy_classical_shadow}

Our assumption thus far has been that the classical shadow protocol is implemented perfectly without being affected by noise. However, this is an unrealistic scenario, as noise is unavoidable in real-world experiments. For classical shadows to be useful even in the presence of noise, it is necessary to employ error mitigation techniques. This is especially important in the present noisy intermediate-scale quantum era \cite{preskill2018quantum,bharti2021noisy}, where quantum devices are noisy and quantum algorithms are performed without quantum error-correction \cite{arute2019quantum}. Fortunately, as previous works have shown \cite{koh2020classical, chen2021robust}, the classical shadow protocol can be made noise-resilient by making modifications to the shadow channel.

In this subsection, let us consider noisy classical shadows with a Pauli-invariant unitary ensemble. Similar to the setting in \cite{koh2020classical, chen2021robust}, we shall assume that a noise channel $\Lambda$ acts on the pre-measurement state $U\rho U^\dag$ just before the measurement is performed. Such an assumption is obeyed by gate-independent, time-stationary, and Markovian noise \cite{chen2021robust, flammia2020efficient}. In this case, the classical shadow is still defined as
$\mathcal{M}^{-1}[\hat{\sigma}_{U,\vec b}]$, although the 
probability distribution is now changed to 
$P_{\Lambda}(U,\vec b)=P(U)P_{\Lambda}(\vec b|U)$ where $P_{\Lambda}(\vec b|U)=\Tr{\proj{\vec b}\Lambda[U\rho U^\dag]}$.
Similar to the noiseless case \eqref{eq:ensembleState}, the ensemble of states is given by
\begin{align*}
\mathcal{E}_{\Lambda,\rho}
=\left\{
\left(\hat{\sigma}_{U,\vec b},
P(U,\vec b)
\right)
\right\}_{U,\vec b}.
\end{align*}
Hence, by
taking the average of the post-measurement states, we have 
$
\sigma
=\mathbb{E}_U
\sum_{\vec b}
\hat{\sigma}_{U,\vec b}
\Tr{\proj{\vec b}\Lambda[U\rho U^\dag]}
=\mathcal{M}_{\Lambda}[\rho],
$
where $\mathcal{M}_{\Lambda}$ is the noisy shadow channel, and 
$
\rho=\mathcal{M}^{-1}[\sigma]
=\mathbb{E}_{\hat{\sigma}\in \mathcal{E}_{\Lambda}}
\mathcal{M}^{-1}_{\Lambda}[\hat{\sigma}],
$
where $\mathcal{M}^{-1}_{\Lambda}$ is the noisy reconstruction map. 
We will now provide an explicit form for the noisy shadow channel
$\mathcal{M}_{\Lambda}$ and the noisy reconstruction map
$\mathcal{M}^{-1}_{\Lambda}$. For clarity of presentation, the detailed proof of all the results in this subsection are provided in Appendix \ref{sec:noisy_classical_shadows}.

\begin{thm}
Given a Pauli-invariant unitary ensemble $\mathcal{E}$ and a noise channel $\Lambda$, 
the noisy shadow channel is given by 
\begin{align}
\mathcal{M}_{\Lambda}[\rho]
=\frac{1}{2^n}\sum_{\vec{a}}W_{\mathcal{E}_{\Lambda}}[\vec{a}]
\Tr{\rho P_{\vec{a}}}P_{\vec{a}},
\end{align}
where $W_{\mathcal{E}_{\Lambda}}[\vec{a}]$ is the average Pauli coefficient of the noisy classical shadow, and 
is defined as
\begin{align}
W_{\mathcal{E}_{\Lambda}}[\vec{a}]=\mathbb{E}_b\mathbb{E}_U
\Tr{\hat{\sigma}_{U,b}P_{\vec{a}}}\Tr{U^\dag \Lambda^\dag[\proj{b}] U P_{\vec{a}}}.
\end{align} 
For the Pauli-invariant unitary ensemble, $\mathcal{M}^{-1}_{\Lambda}$ exists if and only if 
$W_{\mathcal{E}_{\Lambda}}[\vec{a}]>0$ for all $\vec{a}$. If it exists, the reconstruction map 
is  defined by
\begin{align}
\mathcal{M}^{-1}_{\Lambda}[P_{\vec{a}}]
=W_{\mathcal{E}_\Lambda}[\vec{a}]^{-1}P_{\vec{a}}.
\end{align}
\end{thm}

Now, let us consider the 
sample complexity of the noisy classical shadow protocol. Similar to the noiseless case, 
the sample complexity of the noisy classical shadow with the input state $\rho$ and observable $O$ is upper bounded by the (squared) noisy shadow norm
$
\norm{O}^2_{\mathcal{E}_{\Lambda},\rho}
=\mathbb{E}_{\hat{\sigma}\in\mathcal{E}_{\Lambda}}
\hat{o}[\hat{\sigma}]^2,
$ 
where the estimator of the observable  is taken to be $\hat{o}=\Tr{O\mathcal{M}^{-1}_{\Lambda}[\hat{\sigma}]}$.

\begin{prop}
If the observable $O$ is taken to be the
Pauli operator $P_{\vec{a}}$, the shadow norm
is equal to 
\begin{align}
\norm{P_{\vec{a}}}^2_{\mathcal{E}_{\Lambda},\rho}
=W_{\mathcal{E}_{\Lambda}}[\vec{a}]^{-2}W^u_{\mathcal{E}_{\Lambda}}[\vec{a}],
\end{align}
where $W^u_{\mathcal{E}_{\Lambda}}[\vec{a}]$ is defined as
\begin{align}
W^u_{\mathcal{E}_{\Lambda}}[\vec{a}]
=\mathbb{E}_U\mathbb{E}_{\vec b} 
\left|\Tr{P_{\vec{a}}\hat{\sigma}_{U,b}}\right|^2\Tr{\proj{b}\Lambda[\mathbb{I}]}.
\end{align}
Hence, if $\Lambda$ is unital, then $W^u_{\mathcal{E}_{\Lambda}}[\vec{a}]=W_{\mathcal{E}}[\vec{a}]$. 
\end{prop}

From the above proposition, we note that $W^u_{\mathcal{E}_{\Lambda}}[\vec{a}]$ depends on only how the noise channel $\Lambda$ acts on the identity operator. In the case where $\Lambda$ is unital, $W^u_{\mathcal{E}_{\Lambda}}[\vec{a}]$ reduces to $W_{\mathcal{E}}[\vec a]$, the analogous quantity in the noiseless case, which suggests a particular role of noise unitality in the analysis of classical shadows.


Similar to the noiseless case, let us define the (squared) average shadow norm over the Pauli group for the noisy classical shadow to be
$
\norm{O}^2_{\mathcal{E}_\Lambda}
=\mathbb{E}_{V\in \mathcal{P}_n}
\norm{O}^2_{\mathcal{E}_{\Lambda,V\rho V^\dag}}
$. The following proposition gives an explicit expression for the average shadow norm.
\begin{prop}
The average shadow norm in the noisy classical 
shadow protocol with the noise channel $\Lambda$
 can be expressed as follows
\begin{align}
\norm{O}^2_{\mathcal{E}_\Lambda}
=\frac{1}{4^n}
\sum_{\vec{a}}W_{\mathcal{E}_{\Lambda}}[\vec{a}]^{-2}W^u_{\mathcal{E}}[\vec{a}]
W_O[\vec{a}].
\end{align}

\end{prop}

Let us now provide an example of noisy classical shadows and demonstrate how the performance of the protocol depends on the noise rate.

\begin{example}[Noisy classical shadows with global depolarizing noise]

Let us consider the case where the noise channel is the global depolarizing channel, defined as 
$
D_p[\cdot]=(1-p)[\cdot]+p\frac{\Tr{\cdot}\mathbb{I}}{2^n}.
$
Since $D_p$ is unital, we have $W^u_{\mathcal{E}_{D_p}}[\vec{a}]=W_{\mathcal{E}}[\vec{a}]$.
Hence, the coefficients of the 
shadow channel may be expressed as
\begin{align}
W_{\mathcal{E}_{D_p}}[\vec{a}]
=(1-p)^{1-\delta_{\vec{a},\vec{0}}}W_{\mathcal{E}}[\vec{a}].
\end{align}
Thus, the
shadow norm for a Pauli operator $P_{\vec{a}}$ obeys the following identity
\begin{align}
\norm{P_{\vec{a}}}^2_{\mathcal{E}_{D_p},\rho}
=(1-p)^{2\delta_{\vec{a},\vec{0}}-2}\norm{P_{\vec{a}}}^2_{\mathcal{E},\rho},
\end{align}
where $\delta_{\vec{a},\vec{0}}$ denotes the Kronecker delta 
function. 
From the above equation, we see that depolarizing noise increases the number of samples needed for expectation estimation, with an increase that is proportional to the noise rate. 

\end{example}

\section{Application of our results}

\subsection{Classical shadow with locally scrambled unitary ensembles}
First, let us apply our results to locally scrambled unitary ensembles. 
A unitary ensemble is said to be \textit{locally scrambled}  \cite{hu2021classical} if the probability distribution
$P(U)$ satisfies local basis invariance, that is,
\begin{align}
P(U)=P(UV),\quad \forall \mbox{ unitaries } V=V_1\ot\cdots\ot V_n.
\end{align}
It is easy to see that the Pauli-invariance assumption is weaker than the locally scrambled assumption, as
locally scrambled unitary ensembles are Pauli-invariant. Not all Pauli-invariant 
unitary ensembles are locally scrambled though---a counterexample is the Pauli group.

Classical shadows with locally scrambled unitary ensembles were previously considered in \cite{hu2021classical}, which expressed the reconstruction map as the solution of a linear system of size $O(2^n)$; an explicit formula for the reconstruction map was not provided. In our study, however, since locally scrambled unitary ensembles are a special case of the Pauli-invariant unitary ensembles, a consequence of Theorem \ref{thm:main1} is an explicit formula 
for the reconstruction map that circumvents the need to solve an exponential-sized linear system. Our next proposition, whose proof we provide in Appendix \ref{apen:local_scram}, makes this explicit. 

\begin{prop}
Given a locally scrambled unitary ensemble $\mathcal{E}$, the shadow channel is 
\begin{align*}
\mathcal{M}[\rho]
=\frac{1}{2^{n}}\sum_{S\subset [n]}
\bar{W}_{\mathcal{E}}[S]\sum_{\vec{a}:\supp(\vec{a})=S}\Tr{\rho P_{\vec{a}}}P_{\vec{a}},
\end{align*}
where $\bar{W}_{\mathcal{E}}[S]$ is defined as 
$
\bar{W}_{\mathcal{E}}[S]=\mathbb{E}_{\vec{a}:\supp(\vec{a})=S}
W_{\mathcal{E}}[\vec{a}],
$
and the reconstruction map 
is given by
\begin{align*}
\mathcal{M}^{-1}[P_{\vec{a}}]
=\bar{W}_{\mathcal{E}}[\supp(\vec{a})]^{-1}P_{\vec{a}}.
\end{align*}
\end{prop}

Our expression for the reconstruction map above may be compared with that given in \cite{hu2021classical}, which expressed the reconstruction map as 
 $\mathcal{M}^{-1}[\sigma]=\sum_{S\subset [n]}r_SD^{S}[\sigma]$, where $D^S$ denotes the $|S|$-fold Kronecker product of the single-qubit erasure channel $D[\cdot]=\Tr{\cdot}\mathbb{I}/2$ with itself acting on all the qubits indexed by $S$, and the coefficients 
 $r_s$ are given as the solution of a linear system whose coefficients are the entanglement features. In the next proposition, we find an explicit formula for $r_S$ in terms of the coefficients $\bar{W}_{\mathcal{E}}[S]$. 
\begin{prop}\label{prop:r_s}
Given a reconstruction map written as $\mathcal{M}^{-1}[\sigma]=\sum_Sr_SD^{S}[\sigma]$, the coefficients 
$r_S$ can be expressed in terms of $\bar{W}_{\mathcal{E}}[A]$ as follows
\begin{align}
r_S=\sum_{A\subset S}
(-1)^{|S|-|A|}
\bar{W}_{\mathcal{E}}[A^c]^{-1},
\end{align}
where $A^c$ denotes the complement of $A$ in $[n]$.
\end{prop}

The proof of Proposition \ref{prop:r_s} is presented in 
Appendix \ref{apen:local_scram}. 
Using Proposition \ref{prop:entfeature} and the fact that 
$
W_{\mathcal{E}}[S]=3^{|S|}\bar{W}_{\mathcal{E}}[S]
$, our next corollary shows how the coefficients $r_S$ can be expressed in terms of the 
entanglement feature. 

\begin{cor}
Given a reconstruction map written as $\mathcal{M}^{-1}[\sigma]=\sum_Sr_SD^{S}[\sigma]$, the coefficients
$r_S$ can be expressed as follows
\begin{align}
\nonumber r_S=(-1)^{n+|S|}\sum_{A\subset S}
3^{|A^c|}
\left[
\sum_{B\subset A^c}
(-2)^{|B|}
E^{(2)}_{\mathcal{E}}[B]\right]^{-1}.\\
\end{align}
\end{cor}

For the locally scrambled unitary ensembles, the average shadow norm is defined in \cite{hu2021classical} as
$
\norm{O}^2_{\mathcal{E}}
=\mathbb{E}_{V\in U(2)^n}
\norm{O}^2_{\mathcal{E}_{V\rho V^\dag}}
$, which provides a typical lower bound for the sample complexity of expectation estimation using classical shadows with locally scrambled unitary ensembles. 

Just as we have provided an explicit formula for the reconstruction map in terms of the entanglement feature, so too can we provide an explicit formula for the shadow norm 
using the entanglement feature. 

\begin{prop}\label{prop:com_scr}
Given a locally scrambled unitary ensemble $\mathcal{E}$, the 
average shadow norm is
\begin{align}
\norm{O}^2_{\mathcal{E}}
=\frac{1}{4^n}\sum_{S\subset [n]}\bar{W}_{\mathcal{E}}[S]^{-1}
W_O[S],
\end{align}
where $W_O[S]=\sum_{\vec{a}:\supp(\vec{a})=S}W_O[\vec{a}]$.

\end{prop}
The proof of Proposition \ref{prop:com_scr} is presented in Appendix \ref{apen:local_scram}. 
Based on Propositions \ref{prop:entfeature} and \ref{prop:com_scr}, we can express the average shadow norm 
in terms of the entanglement feature as follows:

\begin{align}
\nonumber\norm{O}^2_{\mathcal{E}}
=\frac{1}{4^n}\sum_{S\subset [n]}(-3)^{|S|}
W_O[S]\left[ \sum_{A\subset [S]}
(-2)^{|A|}
E^{(2)}_{\mathcal{E}}[A]
\right]^{-1}.\\
\end{align} 

In summary, this section saw an application of our results to classical shadows with
locally scrambled unitary ensembles, where we
obtained an explicit formula
for the reconstruction map and the average shadow norm in terms of the entanglement features, thus
circumventing the need to solve the exponential-sized system of linear equations in \cite{hu2021classical}.
These explicit formulae may be helpful for the further analysis of the role of entanglement in the classical shadow protocol, which we shall leave for future work.

\subsection{Classical shadow tomography for quantum channels with  Pauli-invariant unitary ensembles}

The classical shadow paradigm was recently extended to the shadow tomography of quantum channels \cite{levy2021classical,kunjummen2021shadow}.
In these works, the unitary ensembles considered were the local and global Clifford ensembles. In this study, we extend their results to the case of Pauli-invariant unitary ensembles. 

Next, we shall outline the procedure proposed in \cite{levy2021classical,kunjummen2021shadow} for constructing the classical shadows for quantum channels. The unitary ensembles used will be assumed to be Pauli-invariant.

(1) Prepare $\ket{\vec{b}_{i}}$ with $\vec{b}_{i}\in\set{0,1}^n$ chosen uniformly randomly.

(2) Apply a unitary $U_i$ chosen from the locally Pauli-invariant ensemble $\mathcal{E}_i$.

(3) Apply the quantum channel $\mathcal{T}$.

(4) Apply a unitary $U$ chosen form the locally Pauli-invariant ensemble $\mathcal{E}_o$, where the unitary ensemble 
$\mathcal{E}_o$ may be different from  $\mathcal{E}_i$.

(5) Measure in the Pauli Z basis to get the output $\vec{b}_o\in\set{0,1}^n$.

The post-measurement state 
is given by  $\hat{\sigma}_{i,o}=U^{T}_{i}\proj{\vec{b}_i}U^{*}_{i}\ot  U^\dag_o\proj{\vec{b}_o} U_o=\hat{\sigma}_i\ot \hat{\sigma}_o$
where $\hat{\sigma}_i=U^{T}_{i}\proj{\vec{b}_i}U^{*}_{i}$ and $ \hat{\sigma}_o=U^\dag_o\proj{\vec{b}_o} U_o$.
Hence, given $\vec{b}_i,U_i,U_o$, the probability 
of getting the outcome $\vec{b}_o$ is given by
$
P(\vec{b}_o|\vec{b}_i,U_i,U_o)
=2^n\Tr{\proj{z}\mathcal{J}(\mathcal{T})},
$
where $\mathcal{J}(\mathcal{T})=\mathbb{I}\ot \Lambda(\proj{\Psi})$ is the Choi-Jamiolkowski state of $\mathcal{T}$, where $\ket{\Psi}=\frac{1}{\sqrt{2^n}}\sum_{\vec{i}}\ket{\vec{i}}\ket{\vec{i}}$ is the Bell state. 
It is easy to verify that 
$\sum_{\vec{b}_o}P(\vec{b}_o|\vec{b}_i,U_i,U_o)
=2^n\Tr{\hat{\sigma}_i\ot \mathbb{I}\mathcal{J}(\mathcal{T})}=1$. 
Also, the probability of obtaining $\vec{b}_i,U_i,U_o$ is equal to 
$
P(\vec{b}_i,U_i,U_o)=P(\vec{b}_i)P(U_i)P(U_o)
$.
Hence,
the ensemble of states is described by 
\begin{align}\label{eq:sha_chan}
\mathcal{E}_{i,o}
=\set{(\hat{\sigma}_{i,o}, P(\vec{b}_i,U_i,U_o,\vec{b}_o))},
\end{align}
where  the probability distribution $P(\vec{b}_i,U_i,U_o,\vec{b}_o)=P(\vec{b}_i,U_i,U_o)P(\vec{b}_o|\vec{b}_i,U_i,U_o)$.
Taking the average of the classical shadow, we obtain 
$
\sigma
=\mathbb{E}_{\hat{\sigma}\in \mathcal{E}_{i,o}}\hat{\sigma}
=\mathcal{M}_{i,o}[\mathcal{J}(\mathcal{T})].
$
To implement classical shadow tomography, the inverse 
of the shadow channel $\mathcal{M}^{-1}_{i,o}$ would need to be implemented. If the inverse exists, then
$
\mathcal{J}(\mathcal{T})=
\mathcal{M}^{-1}_{i,o}[\sigma]
=\mathbb{E}_{\hat{\sigma}\in \mathcal{E}_{i,o}}\mathcal{M}^{-1}_{i,o}[\hat{\sigma}].
$
Given a quantum state $\rho$ and an observable $O$, the estimator $\hat{o}$ is defined as 
$
\hat{o}=
\Tr{\mathcal{M}^{-1}_{i,o}[\hat{\sigma}_{i,o}]\rho^T\ot O}.
$
Therefore,
$
\mathbb{E}[\hat{o}]
=\Tr{\mathcal{T}[\rho]O}.
$

\begin{prop}\label{thm:main_chan}
Given two Pauli-invariant unitary ensembles $\mathcal{E}_{i}$ and $\mathcal{E}_o$, the shadow channel $\mathcal{M}_{i,o}$ is
\begin{align}
\mathcal{M}_{i,o}[P_{\vec{a}_i}\ot P_{\vec{a}_o}]
=
W_{\mathcal{E}_i}[P_{\vec{a}_i}]
W_{\mathcal{E}_o}[P_{\vec{a}_o}]
P_{\vec{a}_i}\ot P_{\vec{a}_o},
\end{align}
where $W_{\mathcal{E}_i}[P_{\vec{a}_i}]=\mathbb{E}_{b_i}\mathbb{E}_{U_i}
|\Tr{\hat{\sigma}_{i}P_{\vec{a}_i}}|^2$ and 
 $W_{\mathcal{E}_o}[P_{\vec{a}_o}]=\mathbb{E}_{b_o}\mathbb{E}_{U_o}
|\Tr{\hat{\sigma}_{o}P_{\vec{a}_o}}|^2$.  That is, $\mathcal{M}_{i,o}=\mathcal{M}_{i}\ot \mathcal{M}_{o}$, where 
$\mathcal{M}_{i}[P_{\vec{a}_i}]=W_{\mathcal{E}_i}[P_{\vec{a}_i}]P_{\vec{a}_i}$ and 
$\mathcal{M}_{o}[P_{\vec{a}_o}]=W_{\mathcal{E}_o}[P_{\vec{a}_o}]P_{\vec{a}_o}$
Hence, for the Pauli-invariant unitary ensemble, $\mathcal{M}^{-1}_{i,o}$ exists iff 
$W_{\mathcal{E}_i}[\vec{a}_i]>0, W_{\mathcal{E}_o}[\vec{a}_o]>0$ for all $\vec{a}$, and the reconstruction map 
is  defined by
\begin{align}
\mathcal{M}^{-1}_{i,o}[P_{\vec{a}_i}\ot P_{\vec{a}_o}]
=W_{\mathcal{E}_i}[\vec{a}_i]^{-1}W_{\mathcal{E}_o}[\vec{a}_o]^{-1}P_{\vec{a}_i}\ot P_{\vec{a}_o}.
\end{align}

\end{prop}
The proof of Proposition \ref{thm:main_chan} is presented in Appendix~\ref{Appendix:ChannelProofs}.

Now, for the classical shadow tomography 
for a quantum channel, the sample complexity is upper bounded by the 
following shadow norm:
\begin{align}
\nonumber\norm{\rho^T\ot O}_{\mathcal{E}_{\mathcal{T}}}^2
=\mathbb{E}_{\mathcal{E}_{\mathcal{T}}}
\hat{o}[\hat{\sigma}_{i,o}]^2.
\end{align}
Next, let us consider the shadow norm in the case where
the observable is some Pauli operator and apply this result to the problem of estimating Pauli channels.

\begin{prop}\label{prop:sam_cha}
If the observable is taken to be a Pauli operator $P_{\vec{a}}$, then the shadow norm is equal to
\begin{align}
\norm{\rho^T\ot P_{\vec{a}}}_{\mathcal{E}_{\mathcal{T}}}
=W_{\mathcal{E}_o}[\vec{a}]^{-1}
\frac{1}{2^{2n}}\sum_{\vec{a}_i}W_{\mathcal{E}_i}[\vec{a}_i]^{-1}
W_{\rho}[\vec{a}_i].
\end{align}
\end{prop}
The proof of Proposition \ref{prop:sam_cha} is presented in Appendix~\ref{Appendix:ChannelProofs}.
Now, we provide an example of a class of channels for which our results on the classical shadow tomography for quantum channels can be applied.

\begin{example}[Estimation of Pauli channels]

A quantum channel $\mathcal{T}$ is called a \textit{Pauli channel} if it can be written as 
$\mathcal{T}[\cdot]=\sum_{\vec{a}\in V^n}p_{\vec{a}}P_{\vec{a}}[\ \cdot\ ]P^\dag_{\vec{a}}$ with $\sum_{\vec{a}}p_{\vec{a}}=1$. Equivalently, Pauli channels $\mathcal{T}$ are precisely those that can be written as
$\mathcal{T}[\cdot]=1/2^n\sum_{\vec{b}}\lambda_{\vec{b}}\Tr{\ \cdot\  P_{\vec{b}}}P_{\vec{b}}$. 
It is easy to see that the probability distribution $\set{p_{\vec{a}}}$ and the coefficients $\set{\lambda_{\vec{b}}}$ are related as follows
$
\lambda_{\vec{b}}
=\frac{1}{2^n}\sum_{\vec{a}}
p_{\vec{a}}(-1)^{\inner{\vec{a}}{\vec{b}}_s}
$, where $\inner{\vec{a}}{\vec{b}}_s=\sum_i\inner{a_i}{b_i}_s$ and 
$\inner{a_i}{b_i}_s=\inner{(x_1,x_2)}{(y_1,y_2)}_s=x_1y_2-x_2y_1$ for any $a_i=(x_1,x_2)$, $b_i=(y_1,y_2)\in V$.
The task of estimating the coefficients $\set{\lambda_{\vec{b}}}_{\vec{b}}$ of the
Pauli channel has been investigated in \cite{fujiwara2003quantum,hayashi2010quantum,chiuri2011experimental,ruppert2012optimal,collins2013mixed,flammia2020efficient,harper2021fast,flammia2021pauli,chen2020quantum}. 
Here, we consider the classical shadow protocol for estimating coefficients $\lambda_{\vec{b}}$ with Pauli-invariant unitary ensembles
$\mathcal{E}_i$ and $\mathcal{E}_o$. The classical shadow is taken to be $\set{\mathcal{M}^{-1}_{i,o}[\hat{\sigma}_{i,o}]}$, and 
the reconstruction map is given in Proposition \ref{thm:main_chan}. To estimate the coefficients $\lambda_{\vec{b}}$, let us consider the observable $P_{\vec{b}}\ot P_{\vec{b}}$, and take the estimator of the observable to be 
$\hat{o}_{\vec b}=
\Tr{\mathcal{M}^{-1}_{i,o}[\hat{\sigma}_{i,o}]P_{\vec{b}}\ot P_{\vec{b}}}$, where it is easy to verify that
$\mathbb{E}[\hat{o}_{\vec b}]=\Tr{\mathcal{J}[\mathcal{T}]P_{\vec{b}}\ot P_{\vec{b}}}=\lambda_{\vec{b}}$.
Then the
sample complexity $\mathscr S$ needed  to accurately predict a collection of $4^n$ linear 
target functions $\left\{\lambda_{\vec{b}}=\Tr{\mathcal{J}[\mathcal{T}]P_{\vec{b}}\ot P_{\vec{b}}]}\right\}_{\vec{b}\in V^n}$ with error $\epsilon$ and failure probability $\delta$ is 
\begin{align*}
    \mathscr S = O\left(
    \frac{n+\log(1/\delta)}{\epsilon^2}
    \max_{\vec{b}}
W_{\mathcal{E}_o}[\vec{b}]^{-1}W_{\mathcal{E}_i}[\vec{b}]^{-1}
    \right).
\end{align*}
\end{example}

\section{Conclusion}
In this study, we investigated the classical shadow protocol with Pauli-invariant unitary ensembles. 
First, we provided an explicit formula for the reconstruction map corresponding to the shadow channel and established a connection between the coefficients of the reconstruction map and the entanglement features of the dynamics. Using the shadow norm, we gave explicit sample complexity upper bounds for the estimation task that the classical shadow protocol solves. Finally, we presented two applications of our results. First, we applied our results to locally scrambled unitary ensembles, where we presented explicit formulas for the reconstruction map and the sample complexity bounds. Second, we applied our results to the shadow process tomography of quantum channels with Pauli-invariant unitary ensembles and provided an example where we considered the task of estimating Pauli channels.
These general results pave the way for more efficient or robust protocols for predicting pertinent properties of quantum states, including their fidelity, entanglement entropy and 
quantum fisher information.

\begin{acknowledgments}
K.B. thanks Hong-Ye Hu for the discussion on classical shadows and its connection with quantum entanglement. This work was supported in part by the ARO Grant W911NF-19-1-0302 and the ARO MURI Grant W911NF-20-1-0082.

\end{acknowledgments}

\bibliography{Paushad}

\appendix
\widetext

\newpage

\section{Symplectic Fourier Transform}


Let $V=\mathbb{Z}_2\times \mathbb{Z}_2$. We shall denote the $n$-fold Cartesian product of $V$ by $V^n = V \times V \times \cdots \times V$ and write its elements as $\vec x = (x_1, x_2, \ldots, x_n)$, where each $x_i = (x_{i1}, x_{i2})$ is a two-tuple. Given two functions $f,g:V^n\to \real$, their inner product is defined as
\begin{align}
\inner{f}{g}=\mathbb{E}_{x\in V^n}[f(x)g(x)].
\end{align}
where the expectation is taken over the uniform distribution, i.e.~$\mathbb E_{x \in V^n} = \frac{1}{4^n} \sum_{x \in V^n}$.

Next, let us consider the function $f:V^n\to \real$. Then it is straightforward to show that $f$ can be written as
\begin{align}
f(x)=\sum_{\vec{u}\in V^n}\hat{f}(u)(-1)^{\inner{u}{x}_s},
\end{align}
where $\hat f:V^n \to \mathbb R$ is the symplectic Fourier transform of $f$ defined as
\begin{align}
\hat{f}(u)=\mathbb{E}_{x\in V^n}\left[f(x)(-1)^{\inner{u}{x}_s}\right],
\end{align}
where the expectation is taken over the uniform distribution, i.e.~$\mathbb E_{x \in V^n} = \frac{1}{4^n} \sum_{x \in V^n}$,
and the symplectic inner product between $\inner{\vec{x}}{\vec{y}}_s=\sum_i\inner{x_i}{y_i}_s$ and 
\begin{align}
\inner{x_i}{y_i}_s=\inner{(a_1,a_2)}{(b_1,b_2)}_s=a_1b_2-a_2b_1,
\end{align}
for any $x_i=(a_1,a_2)$, $y_i=(b_1,b_2)\in V$.

\section{Noiseless classical shadows}
\label{sec:noiseless_classical_shadows}

\begin{thm}\label{thm:apn1}[Restatement of Theorem 1]
For a Pauli-invariant unitary ensemble $\mathcal{E}$, the shadow channel can be written as
\begin{align}
\mathcal{M}[\rho]
=\frac{1}{2^{n}}\sum_{\vec{a}\in V^n}
W_{\mathcal{E}}[\vec{a}]\Tr{\rho P_{\vec{a}}}P_{\vec{a}},
\end{align}
where 
$W_{\mathcal{E}}[\vec{a}]$ is the average squared Pauli coefficient of the classical shadow, defined as 

\begin{align}
W_{\mathcal{E}}[\vec{a}]=\mathbb{E}_{\vec b}\mathbb{E}_U
W_{\hat{\sigma}_{U,\vec b}}[\vec{a}],
\end{align}
where $\mathbb{E}_{\vec b}:=\frac{1}{2^n}\sum_{\vec b}$ denotes the average with respect to the uniform distribution.
Hence, for the Pauli-invariant unitary ensemble, the reconstruction map $\mathcal{M}^{-1}$ exists iff 
$W_{\mathcal{E}}[\vec{a}]>0$ for all $\vec{a}$, and the reconstruction map 
is defined as follows
\begin{align}
\mathcal{M}^{-1}[P_{\vec{a}}]
=W_{\mathcal{E}}[\vec{a}]^{-1}P_{\vec{a}}.
\end{align}

\end{thm}

\begin{proof}
Based on the definition of $\mathcal{M}[\rho]$, we have
\begin{align}
\nonumber\mathcal{M}[\rho]
&=\mathbb{E}_U
\sum_{\vec{b}\in\set{0,1}^n}
\hat{\sigma}_{U,\vec b}
\Tr{\hat{\sigma}_{U,\vec b}\rho}\\
\nonumber&=\sum_{\vec b}\mathbb{E}_U
\hat{\sigma}_{U,\vec b}
\Tr{\hat{\sigma}_{U,\vec b}\rho}\\
\label{eq:1.1}&=\frac{1}{2^{2n}}\sum_{\vec{b}\in\set{0,1}^n}\mathbb{E}_U\sum_{\vec{a},\vec{c}\in V^n}
\Tr{\hat{\sigma}_{U,\vec b}P_{\vec{a}}}P_{\vec{a}}
\Tr{\hat{\sigma}_{U,\vec b}P_{\vec{c}}}\Tr{P_{\vec{c}}\rho}\\
\label{eq:1.2}&=\frac{1}{2^{2n}}\sum_{\vec{b}\in\set{0,1}^n}\sum_{\vec{a},\vec{c}\in V^n}
\mathbb{E}_U\mathbb{E}_{\vec{d}\in V^n}
\Tr{P_{\vec{d}}\hat{\sigma}_{U,\vec b}P_{\vec{d}}P_{\vec{a}}}P_{\vec{a}}
\Tr{P_{\vec{d}}\hat{\sigma}_{U,\vec b}P_{\vec{d}}P_{\vec{c}}}\Tr{P_{\vec{c}}\rho}\\
\label{eq:1.3}&=\frac{1}{2^{2n}}\sum_{\vec{b}\in\set{0,1}^n}\sum_{\vec{a},\vec{c}\in V^n}
\mathbb{E}_U\mathbb{E}_{\vec{d}\in V^n}
\Tr{\hat{\sigma}_{U,b}P_{\vec{a}}}P_{\vec{a}}
\Tr{\hat{\sigma}_{U,b}P_{\vec{c}}}\Tr{P_{\vec{c}}\rho}(-1)^{\inner{\vec{d}}{\vec{a}+\vec{c}}_s}\\
\label{eq:1.4}&=\frac{1}{2^{2n}}\sum_{\vec{b}\in\set{0,1}^n}\sum_{\vec{a},\vec{c}\in V^n}
\mathbb{E}_U\delta_{\vec{a},\vec{c}}
\Tr{\hat{\sigma}_{U,\vec b}P_{\vec{a}}}P_{\vec{a}}
\Tr{\hat{\sigma}_{U,\vec b}P_{\vec{c}}}\Tr{P_{\vec{c}}\rho}\\
\nonumber&=\frac{1}{2^n}\sum_{\vec{a}\in V^n}\left[\frac{1}{2^n}\sum_{\vec{b}\in\set{0,1}^n}
\mathbb{E}_U
\Tr{\hat{\sigma}_{U,\vec b}P_{\vec{a}}}^2
\right]
\Tr{P_{\vec{a}}\rho}P_{\vec{a}}\\
&=\frac{1}{2^n}\sum_{\vec{a}\in V^n}
W_{\mathcal{E}}[\vec{a}]
\Tr{P_{\vec{a}}\rho}P_{\vec{a}},
\end{align}
where the equation \eqref{eq:1.1} comes from the Pauli decomposition of the classical 
shadow $\hat{\sigma}_{U,\vec{b}}=\frac{1}{2^n}\sum_{\vec{a}}\Tr{\hat{\sigma}_{U,\vec{b}} P_{\vec{a}}}P_{\vec{a}}$,
the equation \eqref{eq:1.2} comes from the fact that unitary ensemble is 
locally Pauli-invariant, the equation \eqref{eq:1.3}  comes from the fact $P_{\vec{c}}P_{\vec{a}}P_{\vec{c}}=(-1)^{\inner{\vec{a}}{\vec{c}}_s}P_{\vec{a}}$,
 the equation \eqref{eq:1.3}  comes from the fact $\mathbb{E}_{\vec{c}}(-1)^{\inner{\vec{a}}{\vec{c}}_s}=\delta_{\vec{a},\vec{0}}$.
If $W_{\mathcal{E}}(\vec{a})>0$ for all $\vec{a}$, then the reconstruction map 
can be defined as follows
\begin{align}
\mathcal{M}^{-1}[P_{\vec{a}}]
=W_{\mathcal{E}}[\vec{a}]^{-1}P_{\vec{a}}.
\end{align}
\end{proof}

\begin{prop}[Restatement of Proposition 2]
The coefficients of the reconstruction map can be expressed as follows
 \begin{align}
W_{\mathcal{E}}
[S]=(-1)^{|S|}
\sum_{A\subset [S]}
(-2)^{|A|}
E^{(2)}_{\mathcal{E}}[A].
\end{align}
\end{prop}
\begin{proof}
It is easy to verify that 
 \begin{align}
W_{\mathcal{E}}
[S]=(-1)^{|S|}
\sum_{A\subset [S]}
(-2)^{|A|}
E^{(2)}_{\mathcal{E}}[A].
=(-1)^{|S|}
\sum_{A\subset [S]}
(-2)^{|A|}
\mathbb{E}_{\hat{\sigma}}
\Tr{\hat{\sigma}^2_{A}}
\end{align}
is equivalent to the following equation,
\begin{align}
W_{\mathcal{E}}
[S]
=2^{|S|}
\sum_{A\subset S}
(-1)^{|A|}
\mathbb{E}_{\hat{\sigma}}
\norm{D^A(\hat{\sigma}_S)}^2_2
=2^{|S|}
\sum_{A\subset S}
(-1)^{|A|}
\mathbb{E}_{\hat{\sigma}}
\norm{\hat{\sigma}_{S-A}\ot \frac{\mathbb{I}_A}{2^{|A|}}}^2_2,
\end{align}
where the expectation $\mathbb{E}_{\hat{\sigma}}$
denotes the expectation $\mathbb{E}_{U,b}$, and 
$\hat{\sigma}_S$ denotes the reduced state of $\hat{\sigma}$ on the subset 
$A$.

For any $n$-qubit state $\rho$,
the Pauli coefficient 
$W_{\rho}[\vec{a}]=|\Tr{\rho P_{\vec{a}}}|^2$, and 
\begin{align}
W_{\rho}[S]
=\sum_{\vec{a},\supp(\vec{a})=S}|\Tr{\rho P_{\vec{a}}}|^2
=\sum_{\vec{a},\supp(\vec{a})=S}|\Tr{\rho_S P_{\vec{a}}}|^2,
\end{align} 
where $\rho_S$ is the reduced state on the subset S.

Let us take the 1-qubit erasure channel $D[\cdot]=\Tr{\cdot}\mathbb{I}/2$. 
It is easy to see that $D$ satisfies the following conditions:
(1) $D$ is Hermitian, i.e.~$D^\dag=D$; (2) $D$ is idempotent, i.e.~$D^2=D$. For any subset $A\subset [n]$, let us denote $D^A=\ot_{i\in A} D_i$, where
$D_i$ denotes the erasure channel acting on the $i$-th qubit.
Now, let us define the 1-qubit linear map $\mathcal{L}[\cdot]$ as follows 
\begin{align}
\mathcal{L}[\rho]
=\rho-D[\rho].
\end{align}
It is easy to see that $\mathcal{L}$ satisfies the following conditions
(1) $\mathcal{L}^\dag=\mathcal{L}$; (2) $\mathcal{L}^2=\mathcal{L}$. 
For any subset $S\subset [n]$, let us denote $\mathcal{L}^S=\ot_{i\in S} \mathcal{L}_i$, where
$\mathcal{L}_i$ denotes the linear map $\mathcal{L}$ acting on the $i$-th qubit.
Here, we need the following two lemmas---Lemma \ref{eq:dire_1} and \ref{eq:dire_2}---which can be verified directly. 

\begin{lem}\label{eq:dire_1}
The inner product $\inner{\mathcal{L}^S[\rho_S]}{\rho_S}$ can be written as 
\begin{align}
\inner{\mathcal{L}^S[\rho_S]}{\rho_S}
=\frac{1}{2^{|S|}}
\sum_{\vec{a}:\supp(\vec{a})=S}
|\Tr{\rho_SP_{\vec{a}}}|^2,
\end{align}
where the inner product between two operators is defined as 
$\inner{A}{B}=\Tr{A^\dag B}$.
\end{lem}

\begin{lem}\label{eq:dire_2}
The inner product $\inner{\mathcal{L}^S[\rho_S]}{\rho_S}$ can be rewritten as 
\begin{align}
\inner{\mathcal{L}^S[\rho_S]}{\rho_S}
=\sum_{A\subset S}
(-1)^{|A|}\inner{D^A[\rho_S]}{\rho_S}
=\sum_{A\subset S}
(-1)^{|A|}\norm{D^A[\rho_S]}^2_2.
\end{align}

\end{lem}
Therefore, based on Lemmas \ref{eq:dire_1} and \ref{eq:dire_2}, we have 
\begin{align}
W_{\rho}[S]=2^{|S|}
\sum_{A\subset S}
(-1)^{|A|}
\norm{D^A(\rho_S)}^2_2.
\end{align}

Hence, 
\begin{align}
W_{\mathcal{E}}[S]
=\mathbb{E}_U\mathbb{E}_b
W_{\hat{\sigma}_{U,b}}[S]
=2^{|S|}\sum_{A\subset S}
(-1)^{|A|}\mathbb{E}_U\mathbb{E}_b
\norm{D^A(\hat{\sigma_S})}^2_2.
\end{align}


\end{proof}

\begin{prop}[Restatement of Proposition 3]
If the observable $O$ is taken to be 
Pauli operator $P_{\vec{a}}$, then the shadow norm is equal to 
\begin{align}
\norm{P_{\vec{a}}}^2_{\mathcal{E}_{\rho}}
=W_{\mathcal{E}}[\vec{a}]^{-1}.
\end{align}
\end{prop}
\begin{proof}
First,
\begin{align}
\nonumber\norm{P_{\vec{a}}}^2_{\mathcal{E}_{\rho}}
&=\mathbb{E}_{\mathcal{E}_{\rho}}\Tr{\mathcal{M}^{-1}[P_{\vec{a}}]\hat{\sigma}}^2\\
\nonumber&=\mathbb{E}_U\sum_{\vec {b}\in\set{0,1}^n}
\Tr{\mathcal{M}^{-1}[P_{\vec{a}}]\hat{\sigma}}^2
\Tr{\hat{\sigma}_{U,b}\rho}\\
\label{eq:6.1}&=\mathbb{E}_U\sum_{\vec{b}\in\set{0,1}^n}\mathbb{E}_{\vec{c}\in V^n}
\Tr{P_{\vec{c}}\mathcal{M}^{-1}[P_{\vec{a}}]P_{\vec{c}}\hat{\sigma}_{U,b}}^2
\Tr{\hat{\sigma}_{U,b}P_{\vec{c}}\rho P_{\vec{c}}}\\
\label{eq:6.2}&=\mathbb{E}_U\sum_{\vec{b}\in\set{0,1}^n}\mathbb{E}_{\vec{c}\in V^n}
[(-1)^{\inner{\vec{a}}{\vec{c}}_s}\Tr{\mathcal{M}^{-1}[P_{\vec{a}}]\hat{\sigma}_{U,b}}]^2
\Tr{\hat{\sigma}_{U,b}P_{\vec{c}}\rho P_{\vec{c}}}\\
\nonumber&=\mathbb{E}_U\sum_{\vec{b}\in \set{0,1}^n}
\left[\Tr{\mathcal{M}^{-1}[P_{\vec{a}}]\hat{\sigma}_{U,b}}\right]^2
\mathbb{E}_{\vec{c}\in V^n}\Tr{\hat{\sigma}_{U,b}P_{\vec{c}}\rho P_{\vec{c}}}\\
\label{eq:6.3}&=\mathbb{E}_U\sum_{\vec{b}\in\set{0,1}^n}
\left[\Tr{\mathcal{M}^{-1}[P_{\vec{a}}]\hat{\sigma}_{U,b}}\right]^2
\Tr{\hat{\sigma}_{U,b}\frac{\mathbb{I}}{2^n}}\\
\nonumber&=\frac{1}{2^n}\mathbb{E}_U\sum_{\vec{b}\in\set{0,1}^n}
\left[\Tr{\mathcal{M}^{-1}[P_{\vec{a}}]\hat{\sigma}_{U,b}}\right]^2\\
\label{eq:6.4}&=\frac{1}{2^n}\mathbb{E}_U\sum_{\vec{b}\in\set{0,1}^n} W_{\mathcal{E}}[\vec{a}]^{-2}
\left|\Tr{P_{\vec{a}}\hat{\sigma}_{U,\vec{b}}}\right|^2\\
\label{eq:6.5}&=W_{\mathcal{E}}[\vec{a}]^{-1},
\end{align}
where the equation \eqref{eq:6.1} comes from the fact that the unitary ensemble is 
 Pauli-invariant, and the equation \eqref{eq:6.2} comes from the fact that $P_{\vec{c}}P_{\vec{a}}P_{\vec{c}}=(-1)^{\inner{\vec{a}}{\vec{c}}_s}P_{\vec{a}}$,
 the equation \eqref{eq:6.3} comes from the fact that Pauli group is 1-design\footnote{Note that while the Pauli group is a 1-design, it fails to be 2-design \cite{roy2009unitary,webb2016clifford}.}, 
the equation \eqref{eq:6.4}  comes the definition of reconstruction map 
$\mathcal{M}^{-1}[P_{\vec{a}}]=W_{\mathcal{E}}[\vec{a}]^{-1}P_{\vec{a}}$, and 
the equation \eqref{eq:6.5}  comes the definition of 
$W_{\mathcal{E}}(\vec{a})$.
\end{proof}

\begin{prop}\label{prop:apen2}[Restatement of Proposition 4]
The average  shadow norm over the Pauli group
$\norm{O}^2_{\mathcal{E}}$ can be expressed as
\begin{align}
\norm{O}^2_{\mathcal{E}}
=\frac{1}{4^n}
\sum_{\vec{a}\in V^n}W_{\mathcal{E}}[\vec{a}]^{-1}
W_O[\vec{a}],
\end{align}
where  $W_O[\vec{a}]=|\Tr{OP_{\vec{a}}}|^2$.
\end{prop}
\begin{proof}
Based on the definition of average shadow norm, we have
\begin{align*}
\norm{O}^2_{\mathcal{E}}
&=\mathbb{E}_{V\in \mathcal{P}_n}
\norm{O}^2_{\mathcal{E}_{V\rho V^\dag}}\\
&=\mathbb{E}_{V\in \mathcal{P}_n}
\norm{O}^2_{\mathcal{E}_{\mathbb{I}/2^n}}\\
&=\mathbb{E}_U
\sum_{\vec b}
\Tr{\mathcal{M}^{-1}[O]\hat{\sigma}}^2
\Tr{\hat{\sigma}_{U,\vec b}\frac{\mathbb{I}}{2^n}}\\
&=\mathbb{E}_U
\sum_{\vec b}
\Tr{\mathcal{M}^{-1}[O]\hat{\sigma}}^2
\Tr{\hat{\sigma}_{U,\vec b}\frac{\mathbb{I}}{2^n}}\\
&=\mathbb{E}_U
\mathbb{E}_{\vec b}
\Tr{\mathcal{M}^{-1}[O]\hat{\sigma}}^2\\
&=\mathbb{E}_U
\mathbb{E}_{\vec b}\frac{1}{4^n}
\sum_{\vec{a},\vec{c}}
\Tr{\mathcal{M}^{-1}[O]P_{\vec{a}}}\Tr{P_{\vec{a}}\hat{\sigma}}
\Tr{\mathcal{M}^{-1}[O]P_{\vec{c}}}\Tr{P_{\vec{c}}\hat{\sigma}}\\
&=\mathbb{E}_U
\mathbb{E}_{\vec b}\frac{1}{4^n}
\sum_{\vec{a},\vec{c}}
\Tr{\mathcal{M}^{-1}[O]P_{\vec{a}}}\Tr{P_{\vec{a}}\hat{\sigma}}
\Tr{\mathcal{M}^{-1}[O]P_{\vec{c}}}\Tr{P_{\vec{c}}\hat{\sigma}}
\delta_{\vec{a},\vec{c}}\\
&=\mathbb{E}_U
\mathbb{E}_{\vec b}\frac{1}{4^n}
\sum_{\vec{a}}
|\Tr{\mathcal{M}^{-1}[O]P_{\vec{a}}}|^2|\Tr{P_{\vec{a}}\hat{\sigma}}|^2\\
&=\frac{1}{4^n}
\sum_{\vec{a}}
|\Tr{\mathcal{M}^{-1}[O]P_{\vec{a}}}|^2 W_{\mathcal{E}}[\vec{a}]\\
&=\frac{1}{4^n}
\sum_{\vec{a}}W_{\mathcal{E}}[\vec{a}]^{-1}
|\Tr{OP_{\vec{a}}}|^2\\
&=\frac{1}{4^n}
\sum_{\vec{a}}W_{\mathcal{E}}[\vec{a}]^{-1}
W_{O}[\vec{a}].
\end{align*}
\end{proof}

\section{Noisy classical shadows}
\label{sec:noisy_classical_shadows}

\begin{thm}[Restatement of Theorem 5]
Given a Pauli-invariant unitary ensemble $\mathcal{E}$ and a noise channel $\Lambda$, 
then shadow channel is given by 
\begin{align}
\mathcal{M}_{\Lambda}[\rho]
=\frac{1}{2^n}\sum_{\vec{a}}W_{\mathcal{E}_{\Lambda}}[\vec{a}]
\Tr{\rho P_{\vec{a}}}P_{\vec{a}},
\end{align}
where $W_{\mathcal{E}_{\Lambda}}[\vec{a}]$ is the average Pauli coefficient of the noisy classical shadow, and 
is defined as follows
\begin{align}
W_{\mathcal{E}_{\Lambda}}[\vec{a}]=\mathbb{E}_b\mathbb{E}_U
\Tr{\hat{\sigma}_{U,b}P_{\vec{a}}}\Tr{U^\dag \Lambda^\dag[\proj{b}] U P_{\vec{a}}}.
\end{align} 
Hence, for the Pauli-invariant unitary ensemble, $\mathcal{M}^{-1}_{\Lambda}$ exists iff 
$W_{\mathcal{E}_{\Lambda}}[\vec{a}]>0$ for all $\vec{a}$, and the reconstruction map 
is  defined as follows
\begin{align}
\mathcal{M}^{-1}_{\Lambda}[P_{\vec{a}}]
=W_{\mathcal{E}_\Lambda}[\vec{a}]^{-1}P_{\vec{a}}.
\end{align}
\end{thm}
\begin{proof}
First,
\begin{align*}
\mathcal{M}[\rho]
&=\mathbb{E}_U
\sum_{\vec{b}\in\set{0,1}^n}
\hat{\sigma}_{U,\vec b}
\Tr{\proj{\vec b}\Lambda[U\rho U^\dag]}\\
&=\frac{1}{2^{2n}}
\sum_{\vec{b}\in\set{0,1}^n}\mathbb{E}_U
\sum_{\vec{a},\vec{c}\in V^n}
\Tr{\hat{\sigma}_{U,\vec b}P_{\vec{a}}}
P_{\vec{a}}
\Tr{U^\dag \Lambda^\dag\left[\proj{\vec b}\right] U P_{\vec{c}}]}
\Tr{\rho P_{\vec{c}}}\\
&=\frac{1}{2^{2n}}
\sum_{\vec{b}\in\set{0,1}^n}\mathbb{E}_U
\sum_{\vec{a},\vec{c}\in V^n}\mathbb{E}_{\vec{d}}
\Tr{\hat{\sigma}_{U,\vec b}P_{\vec{a}}}
P_{\vec{a}}
\Tr{U^\dag \Lambda^\dag\left[\proj{\vec b}\right] U P_{\vec{c}}]}
\Tr{\rho P_{\vec{c}}}(-1)^{\inner{\vec{d}}{\vec{a}+\vec{c}}_s}\\
&=\frac{1}{2^{2n}}
\sum_{\vec{b}\in\set{0,1}^n}\mathbb{E}_U
\sum_{\vec{a},\vec{c}\in V^n}
\Tr{\hat{\sigma}_{U,b}P_{\vec{a}}}
P_{\vec{a}}
\Tr{U^\dag \Lambda^\dag\left[\proj{\vec b}\right] U P_{\vec{c}}}
\Tr{\rho P_{\vec{c}}}\delta_{\vec{a},\vec{c}}\\
&=\frac{1}{2^{2n}}
\sum_{\vec{b}\in\set{0,1}^n}\mathbb{E}_U
\sum_{\vec{a}\in V^n}
\Tr{\hat{\sigma}_{U,\vec b}P_{\vec{a}}}
P_{\vec{a}}
\Tr{U^\dag \Lambda^\dag\left[\proj{\vec b}\right] U P_{\vec{a}}]}
\Tr{\rho P_{\vec{a}}}\\
&=\frac{1}{2^n}\sum_{\vec{a}\in V^n}W_{\mathcal{E}_{\Lambda}}[\vec{a}]
\Tr{\rho P_{\vec{a}}}P_{\vec{a}},
\end{align*}
where $W_{\mathcal{E}_{\Lambda}}[\vec{a}]$ is defined as 
\begin{align}
W_{\mathcal{E}_{\Lambda}}[\vec{a}]=\mathbb{E}_{\vec b}\mathbb{E}_U
\Tr{\hat{\sigma}_{U,b}P_{\vec{a}}}\Tr{U^\dag \Lambda^\dag[\proj{b}] U P_{\vec{a}}}.
\end{align}

\end{proof}

\begin{prop}[Restatement of Proposition 6]
If the observable $O$ is taken to be 
Pauli operator $P_{\vec{a}}$, then the shadow norm
is equal to 
\begin{align}
\norm{P_{\vec{a}}}^2_{\mathcal{E}_{\Lambda},\rho}
=W_{\mathcal{E}_{\Lambda}}[\vec{a}]^{-2}W^u_{\mathcal{E}_{\Lambda}}[\vec{a}],
\end{align}
where $W^u_{\mathcal{E}_{\Lambda}}[\vec{a}]$ is defined as
\begin{align}
W^u_{\mathcal{E}_{\Lambda}}[\vec{a}]
=\mathbb{E}_U\mathbb{E}_{\vec b} 
|\Tr{P_{\vec{a}}\hat{\sigma}_{U,b}}|^2\Tr{\proj{b}\Lambda[\mathbb{I}]}.
\end{align}
Hence, if $\Lambda$ is unital, then $W^u_{\mathcal{E}_{\Lambda}}[\vec{a}]=W_{\mathcal{E}}[\vec{a}]$. 
\end{prop}
\begin{proof}
First,
\begin{align}
\nonumber\norm{P_{\vec{a}}}^2_{\mathcal{E}_{\Lambda},\rho}
&=\mathbb{E}_{\mathcal{E}_{\Lambda}}\Tr{\mathcal{M}^{-1}_{\Lambda}[P_{\vec{a}}]\hat{\sigma}}^2\\
\nonumber&=\mathbb{E}_U\sum_{\vec{b}\in\set{0,1}^n}
\Tr{\mathcal{M}^{-1}_{\Lambda}[P_{\vec{a}}]\hat{\sigma}}^2
\Tr{U^\dag\Lambda^\dag\left[\proj{\vec b}\right]U\rho}\\
\label{eq:6.1x}&=\mathbb{E}_U\sum_{\vec{b}\in\set{0,1}^n}\mathbb{E}_{\vec{c}\in V^n}
\Tr{P_{\vec{c}}\mathcal{M}^{-1}_{\Lambda}[P_{\vec{a}}]P_{\vec{c}}\hat{\sigma}_{U,\vec b}}^2
\Tr{U^\dag\Lambda^\dag\left[\proj{\vec b}\right]UP_{\vec{c}}\rho P_{\vec{c}}}\\
\label{eq:6.2x}&=\mathbb{E}_U\sum_{\vec{b}\in\set{0,1}^n}\mathbb{E}_{\vec{c}\in V^n}
[(-1)^{\inner{\vec{a}}{\vec{c}}}\Tr{\mathcal{M}^{-1}_{\Lambda}[P_{\vec{a}}]\hat{\sigma}_{U,\vec b}}^2
\Tr{U^\dag\Lambda^\dag[\proj{\vec b}]UP_{\vec{c}}\rho P_{\vec{c}}}\\
\nonumber&=\mathbb{E}_U\sum_{\vec{b}\in\set{0,1}^n}
\left[\Tr{\mathcal{M}^{-1}_{\Lambda}[P_{\vec{a}}]\hat{\sigma}_{U,b}}\right]^2
\mathbb{E}_{\vec{c}\in V^n}\Tr{U^\dag\Lambda^\dag[\proj{b}]UP_{\vec{c}}\rho P_{\vec{c}}}\\
\label{eq:6.3x}&=\mathbb{E}_U\sum_{\vec{b}\in\set{0,1}^n}
\left[\Tr{\mathcal{M}^{-1}_{\Lambda}[P_{\vec{a}}]\hat{\sigma}_{U,\vec b}}\right]^2
\Tr{U^\dag\Lambda^\dag[\proj{\vec b}]U\frac{\mathbb{I}}{2^n}}\\
\nonumber&=\frac{1}{2^n}\mathbb{E}_U\sum_{\vec{b}\in\set{0,1}^n}
\left[\Tr{\mathcal{M}^{-1}_{\Lambda}[P_{\vec{a}}]\hat{\sigma}_{U,b}}\right]^2\Tr{\proj{\vec b}\Lambda[\mathbb{I}]}\\
\label{eq:6.4x}&=\frac{1}{2^n}\mathbb{E}_U\sum_{\vec{b}\in\set{0,1}^n} W_{\mathcal{E}}[\vec{a}]^{-2}
|\Tr{P_{\vec{a}}\hat{\sigma}_{U,\vec b}}|^2\Tr{\proj{\vec b}\Lambda[\mathbb{I}]}\\
\label{eq:6.5x}&=W_{\mathcal{E}_{\Lambda}}[\vec{a}]^{-2}W^u_{\mathcal{E}_{\Lambda}}[\vec{a}],
\end{align}
where 
\begin{align}
W^u_{\mathcal{E}_{\Lambda}}[\vec{a}]
=\frac{1}{2^n}\mathbb{E}_U\sum_{\vec b} 
|\Tr{P_{\vec{a}}\hat{\sigma}_{U,b}}|^2\Tr{\proj{b}\Lambda[\mathbb{I}]}.
\end{align}
\end{proof}

\begin{prop}[Restatement of Proposition 7]
The average shadow norm in the noisy classical 
shadow protocol with the noise channel $\Lambda$
 can be expressed as follows
\begin{align}
\norm{O}^2_{\mathcal{E}_\Lambda}
=\frac{1}{4^n}
\sum_{\vec{a}}W_{\mathcal{E}_{\Lambda}}[\vec{a}]^{-2}W^u_{\mathcal{E}}[\vec{a}]
W_O[\vec{a}].
\end{align}

\end{prop}
\begin{proof}
First,
\begin{align*}
\norm{O}^2_{\mathcal{E}}
&=\mathbb{E}_{V\in\mathcal{P}_n}
\norm{O}^2_{\mathcal{E}_{V\rho V^\dag}}\\
&=
\norm{O}^2_{\mathcal{E}_{\mathbb{I}/2^n}}\\
&=\mathbb{E}_U
\sum_{\vec b}
\Tr{\mathcal{M}^{-1}_{\Lambda}[O]\hat{\sigma}}^2
\Tr{U^\dag\Lambda^\dag[\proj{b}]U\frac{\mathbb{I}}{2^n}}\\
&=\mathbb{E}_U
\sum_{\vec b}
\Tr{\mathcal{M}^{-1}_{\Lambda}[O]\hat{\sigma}}^2
\Tr{U^\dag\Lambda^\dag[\proj{b}]U\frac{\mathbb{I}}{2^n}}\\
&=\mathbb{E}_U
\mathbb{E}_{\vec b}
\Tr{\mathcal{M}^{-1}_{\Lambda}[O]\hat{\sigma}}^2\Tr{\proj{b}\Lambda[\mathbb{I}]}\\
&=\mathbb{E}_U
\mathbb{E}_{\vec b}\frac{1}{4^n}
\sum_{\vec{a}}
|\Tr{\mathcal{M}^{-1}_{\Lambda}[O]P_{\vec{a}}}|^2|\Tr{P_{\vec{a}}\hat{\sigma}}|^2\Tr{\proj{b}\Lambda[\mathbb{I}]}\\
&=\frac{1}{4^n}
\sum_{\vec{a}}
|\Tr{\mathcal{M}^{-1}_{\Lambda}[O]P_{\vec{a}}}|^2 W^u_{\mathcal{E}}[\vec{a}]\\
&=\frac{1}{4^n}
\sum_{\vec{a}}W_{\mathcal{E}_{\Lambda}}[\vec{a}]^{-2}W^u_{\mathcal{E}}[\vec{a}]
|\Tr{OP_{\vec{a}}}|^2.
\end{align*}
\end{proof}

\section{locally scrambled unitary ensembles}\label{apen:local_scram}

\begin{prop}[Restatement of Proposition 8]
Given a locally scrambled unitary ensemble $\mathcal{E}$, then the shadow channel is 
\begin{align*}
\mathcal{M}[\rho]
=\frac{1}{2^{n}}\sum_{S\subset [n]}
\bar{W}_{\mathcal{E}}[S]\sum_{\vec{a}:\supp(\vec{a})=S}\Tr{\rho P_{\vec{a}}}P_{\vec{a}},
\end{align*}
where $\bar{W}_{\mathcal{E}}(S)$ is defined as 
$
\bar{W}_{\mathcal{E}}[S]=\mathbb{E}_{\vec{a}:\supp(\vec{a})=S}
W_{\mathcal{E}}[\vec{a}].
$
And the reconstruction map 
is defined as follows
\begin{align*}
\mathcal{M}^{-1}[P_{\vec{a}}]
=\bar{W}_{\mathcal{E}}[\supp(\vec{a})]^{-1}P_{\vec{a}}.
\end{align*}
\end{prop}
\begin{proof}
Since the locally scrambled unitary ensemble is Pauli-invariant, then according to the above theorem, we have 
\begin{align}
\mathcal{M}[\rho]
=\frac{1}{2^{n}}\sum_{\vec{a}}
W_{\mathcal{E}}[\vec{a}]\Tr{\rho P_{\vec{a}}}P_{\vec{a}},
\end{align}
and 
\begin{align*}
W_{\mathcal{E}}[\vec{a}]&=\frac{1}{2^n}\sum_{\vec b}\mathbb{E}_U
(\Tr{\hat{\sigma}_{U,b}P_{\vec{a}}})^2\\
&=\frac{1}{2^n}\sum_{\vec b}\mathbb{E}_U\mathbb{E}_{V^n}
\left(\Tr{V^n\hat{\sigma}_{U,\vec b}(V^n)^\dag P_{\vec{a}}}\right)^2\\
&=\frac{1}{2^n}\sum_{\vec b}\mathbb{E}_U\mathbb{E}_{\supp(\vec{a})}
\left(\Tr{\hat{\sigma}_{U,\vec b}\mathbb{I}_{\supp(\vec{a})^c}\ot P_{\supp(\vec{a})}}\right)^2,
\end{align*}
where the expectation $\mathbb{E}_{\supp(\vec{a})}$ is taken over all the Pauli operators with support $\supp(\vec{a})$
 without identity component. 
 Hence 
 \begin{align}
 W_{\mathcal{E}}[\vec{a}]=
\mathbb{E}_{\supp(\vec{a})} W_{\mathcal{E}}[\vec{a}] =W_{\mathcal{E}}[\supp(\vec{a})].
 \end{align}

\end{proof}

\begin{prop}[Restatement of Proposition 9]
Given the reconstruction map $\mathcal{M}^{-1}[\sigma]=\sum_Sr_SD^{S}[\sigma]$, then 
$r_S$ can be expressed by 
$\bar{W}_{\mathcal{E}}[A]$ as follows
\begin{align}
r_S=\sum_{A\subset S}
(-1)^{|S|-|A|}
\bar{W}_{\mathcal{E}}[A^c]^{-1},
\end{align}
where $A^c$ denotes the complement of $A$ in $[n]$.
\end{prop}
\begin{proof}
Since $\mathcal{M}^{-1}[\sigma]=\sum_Sr_SD^{S}[\sigma]$, where
$D(\cdot)=\Tr{\cdot}\mathbb{I}/2$ and $D^S=\ot_{i\in S}D_i$. 
Hence 
\begin{align}
\mathcal{M}^{-1}[P_{\vec{a}}]
=\sum_Sr_SD^{S}[P_{\vec{a}}]
=P_{\vec{a}}\sum_Sr_S I[\supp(\vec{a})\subset S^c],
\end{align}
where $I[\cdot]$ is the indicator function.
Therefore, we have
\begin{align}
\sum_{S}r_SI[A\subset S^c]
=W_{\mathcal{E}}[A]^{-1},
\end{align}
which is equivalent to 
\begin{align}
\sum_{S}r_SI[S\subset A]
=W_{\mathcal{E}}[A^c]^{-1}.
\end{align}

Hence, we have 
\begin{align}
r_S=\sum_{A\subset S}
(-1)^{|S|-|A|}
W_{\mathcal{E}}[A^c]^{-1},
\end{align}
which follows from the inclusion-exclusion principle.
\end{proof}

\begin{lem}[Inclusion-exclusion principle]
If $g(A)=\sum_{S\subset A}f(S)$, then we have 
\begin{align}
f(A)=\sum_{S\subset A}
(-1)^{|A|-|S|}
g(S).
\end{align}
\end{lem}

\begin{prop}[Restatement of Proposition 11]
Given a locally scrambled unitary ensemble, the 
average shadow norm is
\begin{align}
\norm{O}^2_{\mathcal{E}}
=\frac{1}{4^n}\sum_{S\subset [n]}\bar{W}_{\mathcal{E}}[S]^{-1}
W_O[S],
\end{align}
where $W_O[S]=\sum_{\vec{a}:\supp(\vec{a})=S}W_O[\vec{a}]$.

\end{prop}
\begin{proof}
Since locally scrambled unitary ensembles are Pauli-invariant, we have
\begin{align}
\norm{O}^2_{\mathcal{E}}
=\frac{1}{4^n}
\sum_{\vec{a}}W_{\mathcal{E}}[\vec{a}]^{-1}
W_{O}[\vec{a}],
\end{align}
and 
\begin{align*}
W_{\mathcal{E}}(\vec{a})&=\frac{1}{2^n}\sum_{b}\mathbb{E}_U
\left(\Tr{\hat{\sigma}_{U,\vec b}P_{\vec{a}}}\right)^2\\
&=\frac{1}{2^n}\sum_{\vec b}\mathbb{E}_U\mathbb{E}_{V^n}
\left(\Tr{V^n\hat{\sigma}_{U,\vec b}(V^n)^\dag P_{\vec{a}}}\right)^2\\
&=\frac{1}{2^n}\sum_{\vec b}\mathbb{E}_U\mathbb{E}_{\supp(\vec{a})}
\left(\Tr{\hat{\sigma}_{U,\vec b}\mathbb{I}_{\supp(\vec{a})^c}\ot P_{\supp(\vec{a})}}\right)^2,
\end{align*} 
where the expectation $\mathbb{E}_{\supp(\vec{a})}$ is taken over all the Pauli operation on support $\supp(\vec{a})$
 without identity component. 
 Hence 
 \begin{align}
 W_{\mathcal{E}}[\vec{a}]=
 \mathbb{E}_{\supp(\vec{a})} W_{\mathcal{E}}[\vec{a}] =W_{\mathcal{E}}[\supp(\vec{a})].
 \end{align}

\end{proof}

\section{Classical shadow channel tomography with Pauli-invariant unitary ensembles}\label{Appendix:ChannelProofs}

\begin{prop}[Restatement of Proposition 12]
Given two Pauli-invariant unitary ensembles $\mathcal{E}_{i}$ and $\mathcal{E}_o$, the shadow channel $\mathcal{M}_{i,o}$ is
\begin{align}
\mathcal{M}_{i,o}[P_{\vec{a}_i}\ot P_{\vec{a}_o}]
=
W_{\mathcal{E}_i}[P_{\vec{a}_i}]
W_{\mathcal{E}_o}[P_{\vec{a}_o}]
P_{\vec{a}_i}\ot P_{\vec{a}_o},
\end{align}
where $W_{\mathcal{E}_i}[P_{\vec{a}_i}]=\mathbb{E}_{b_i}\mathbb{E}_{U_i}
|\Tr{\hat{\sigma}_{i}P_{\vec{a}_i}}|^2$ and 
 $W_{\mathcal{E}_o}[P_{\vec{a}_o}]=\mathbb{E}_{b_o}\mathbb{E}_{U_o}
|\Tr{\hat{\sigma}_{o}P_{\vec{a}_o}}|^2$.  That is, $\mathcal{M}_{i,o}=\mathcal{M}_{i}\ot \mathcal{M}_{o}$, where 
$\mathcal{M}_{i}[P_{\vec{a}_i}]=W_{\mathcal{E}_i}[P_{\vec{a}_i}]P_{\vec{a}_i}$ and 
$\mathcal{M}_{o}[P_{\vec{a}_o}]=W_{\mathcal{E}_o}[P_{\vec{a}_o}]P_{\vec{a}_o}$
Hence, for the Pauli-invariant unitary ensemble, $\mathcal{M}^{-1}_{i,o}$ exists iff 
$W_{\mathcal{E}_i}[\vec{a}_i]>0, W_{\mathcal{E}_o}[\vec{a}_o]>0$ for all $\vec{a}$, and the reconstruction map 
is  defined as follows
\begin{align}
\mathcal{M}^{-1}_{i,o}[P_{\vec{a}_i}\ot P_{\vec{a}_o}]
=W_{\mathcal{E}_i}[\vec{a}_i]^{-1}W_{\mathcal{E}_o}[\vec{a}_o]^{-1}P_{\vec{a}_i}\ot P_{\vec{a}_o}.
\end{align}

\end{prop}

\begin{proof}
The proof is similar to the proof of Theorem \ref{thm:apn1}, but we provide it here for the completeness. 
\begin{align*}
&\mathcal{M}_{i,o}[P_{\vec{a}_i}\ot P_{\vec{a}_o}]\nonumber\\
&=\mathbb{E}_{b_i}\mathbb{E}_{U_i}
\mathbb{E}_{U_o}2^n
\sum_{b_o}\hat{\sigma}_{i,o}\Tr{\hat{\sigma}_{i,o}P_{\vec{a}_i}\ot P_{\vec{a}_o}}\\
&=\frac{1}{2^{2n}}\mathbb{E}_{U_i}
\mathbb{E}_{U_o}
\sum_{b_i,b_o}\sum_{\vec{c}_i,\vec{c}_o}
\Tr{\hat{\sigma}_{i,o}P_{\vec{c}_i}\ot P_{\vec{c}_o}}P_{\vec{c}_i}\ot P_{\vec{c}_o}
\Tr{\hat{\sigma}_{i,o}P_{\vec{a}_i}\ot P_{\vec{a}_o}}\\
&=\frac{1}{2^{2n}}\mathbb{E}_{U_i}
\mathbb{E}_{U_o}
\sum_{b_i,b_o}\sum_{\vec{c}_i,\vec{c}_o}
\mathbb{E}_{\vec{d}_i,\vec{d}_o}
\Tr{P_{\vec{d}_i}\ot P_{\vec{d}_o}\hat{\sigma}_{i,o} P_{\vec{d}_i}\ot P_{\vec{d}_o}P_{\vec{c}_i}\ot P_{\vec{c}_o}}P_{\vec{c}_i}\ot P_{\vec{c}_o}
\Tr{P_{\vec{d}_i}\ot P_{\vec{d}_o}\hat{\sigma}_{i,o}P_{\vec{d}_i}\ot P_{\vec{d}_o}P_{\vec{a}_i}\ot P_{\vec{a}_o}}\\
&=\frac{1}{2^{2n}}\mathbb{E}_{b_i}\mathbb{E}_{U_i}
\mathbb{E}_{U_o}
\sum_{b_i,b_o}\sum_{\vec{c}_i,\vec{c}_o}
\mathbb{E}_{\vec{d}_i,\vec{d}_o}
\Tr{\hat{\sigma}_{i,o}P_{\vec{c}_i}\ot P_{\vec{c}_o}}P_{\vec{c}_i}\ot P_{\vec{c}_o}
\Tr{\hat{\sigma}_{i,o}P_{\vec{a}_i}\ot P_{\vec{a}_o}}(-1)^{\inner{\vec{d}_i}{\vec{a}_i+\vec{c}_i}+\inner{\vec{d}_o}{\vec{a}_o+\vec{c}_o}}\\
&=\frac{1}{2^{2n}}\mathbb{E}_{b_i}\mathbb{E}_{U_i}
\mathbb{E}_{U_o}
\sum_{b_i,b_o}\sum_{\vec{c}_i,\vec{c}_o}
\Tr{\hat{\sigma}_{i,o}P_{\vec{c}_i}\ot P_{\vec{c}_o}}P_{\vec{c}_i}\ot P_{\vec{c}_o}
\Tr{\hat{\sigma}_{i,o}P_{\vec{a}_i}\ot P_{\vec{a}_o}}\delta_{\vec{a}_i,\vec{c}_i}\delta_{\vec{a}_o,\vec{c}_o}\\
&=\frac{1}{2^{2n}}\mathbb{E}_{b_i}\mathbb{E}_{U_i}
\mathbb{E}_{U_o}\sum_{b_i,b_o}
|\Tr{\hat{\sigma}_{i,o}P_{\vec{a}_i}\ot P_{\vec{a}_o}}|^2P_{\vec{a}_i}\ot P_{\vec{a}_o}\\
&=\left[\mathbb{E}_{b_i}\mathbb{E}_{U_i}
|\Tr{\hat{\sigma}_{i}P_{\vec{a}_i}}|^2\right]
\left[\mathbb{E}_{U_o}\mathbb{E}_{b_o}|\Tr{\hat{\sigma}_{i}P_{\vec{a}_o}}|^2\right]
P_{\vec{a}_i}\ot P_{\vec{a}_o}\\
&=
W_{\mathcal{E}_i}[P_{\vec{a}_i}]
W_{\mathcal{E}_o}[P_{\vec{a}_o}]
P_{\vec{a}_i}\ot P_{\vec{a}_o}.
\end{align*}
\end{proof}

\begin{prop}[Restatement of Proposition 13]
If the observable is taken to be a Pauli operator $P_{\vec{a}}$, then the shadow norm is equal to
\begin{align}
\norm{\rho^T\ot P_{\vec{a}}}_{\mathcal{E}_{\mathcal{T}}}
=W_{\mathcal{E}_o}[\vec{a}]^{-1}
\frac{1}{2^{2n}}\sum_{\vec{a}_i}W_{\mathcal{E}_i}[\vec{a}_i]^{-1}
W_{\rho}[\vec{a}_i].
\end{align}
\end{prop}
\begin{proof}
The proof is similar to the proof of Proposition \ref{prop:apen2}, but we present it here for completeness.
\begin{align*}
\norm{\rho\ot P_{\vec{a}}}_{\mathcal{E}_{\mathcal{T}}}
&=\mathbb{E}_{\mathcal{E}}
\Tr{\mathcal{M}^{-1}_{i,o}[\rho^T\ot P_{\vec{a}}] \hat{\sigma}}^2\\
&=
\mathbb{E}_{U_i}\mathbb{E}_{U_o}
\sum_{b_i,b_o}
\Tr{\mathcal{M}^{-1}_{i,o}[\rho^T\ot P_{\vec{a}}] \hat{\sigma}}^2
\Tr{\hat{\sigma}\mathcal{J}[\mathcal{T}]}
\\
&=\mathbb{E}_{U_i}\mathbb{E}_{U_o}
\sum_{b_i,b_o}
|\Tr{\mathcal{M}^{-1}_{i}[\rho^T]\hat{\sigma}_i}|^2|\Tr{\mathcal{M}^{-1}_o[P_{\vec{a}}] \hat{\sigma}_o}|^2
\Tr{\hat{\sigma}_i\ot \hat{\sigma}_o\mathcal{J}[\mathcal{T}]}\\
&=\mathbb{E}_{U_i}\mathbb{E}_{U_o}
\sum_{b_i,b_o}\mathbb{E}_{P_{\vec{c}_o}}
|\Tr{\mathcal{M}^{-1}_{i}[\rho^T]\hat{\sigma}_i}|^2|\Tr{\mathcal{M}^{-1}_o[P_{\vec{a}}]P_{\vec{c}_o} \hat{\sigma}_oP_{\vec{c}_o}}|^2
\Tr{\hat{\sigma}_i\ot P_{\vec{c}_o} \hat{\sigma}_oP_{\vec{c}_o}\mathcal{J}[\mathcal{T}]}\\
&=\mathbb{E}_{U_i}\mathbb{E}_{U_o}
\sum_{b_i,b_o}
|\Tr{\mathcal{M}^{-1}_{i}[\rho^T]\hat{\sigma}_i}|^2|\Tr{\mathcal{M}^{-1}_o[P_{\vec{a}}] \hat{\sigma}_o}|^2
\mathbb{E}_{P_{\vec{c}_o}}\Tr{\hat{\sigma}_i\ot P_{\vec{c}_o} \hat{\sigma}_oP_{\vec{c}_o}\mathcal{J}[\mathcal{T}]}\\
&=\mathbb{E}_{U_i}\mathbb{E}_{U_o}
\sum_{b_i,b_o}
|\Tr{\mathcal{M}^{-1}_{i}[\rho^T]\hat{\sigma}_i}|^2|\Tr{\mathcal{M}^{-1}_o[P_{\vec{a}}] \hat{\sigma}_o}|^2
\Tr{\hat{\sigma}_i\ot \frac{\mathbb{I}}{2^n}\mathcal{J}[\mathcal{T}]}\\
&=\frac{1}{2^{2n}}\mathbb{E}_{U_i}\mathbb{E}_{U_o}
\sum_{b_i,b_o}
|\Tr{\mathcal{M}^{-1}_{i}[\rho^T]\hat{\sigma}_i}|^2|\Tr{\mathcal{M}^{-1}_o[P_{\vec{a}}] \hat{\sigma}_o}|^2\\
&=W_{\mathcal{E}_o}[\vec{a}]^{-1}
\frac{1}{2^n}\mathbb{E}_{U_i}
\sum_{b_i}
|\Tr{\mathcal{M}^{-1}_{i}[\rho^T]\hat{\sigma}_i}|^2\\
&=W_{\mathcal{E}_o}[\vec{a}]^{-1}
\frac{1}{2^{3n}}\mathbb{E}_{U_i}
\sum_{b_i}\sum_{\vec{a}_i,\vec{c}_i}
\Tr{\mathcal{M}^{-1}_{i}[\rho^T]P_{\vec{a}_i}}\Tr{P_{\vec{a}_i}\hat{\sigma}_i}
\Tr{\mathcal{M}^{-1}_{i}[\rho^T]P_{\vec{c}_i}}\Tr{P_{\vec{c}_i}\hat{\sigma}_i}
\\
&=W_{\mathcal{E}_o}[\vec{a}]^{-1}
\frac{1}{2^{3n}}\mathbb{E}_{U_i}
\sum_{b_i}\sum_{\vec{a}_i}
|\Tr{\mathcal{M}^{-1}_{i}[\rho^T]P_{\vec{a}_i}}|^2|\Tr{P_{\vec{a}_i}\hat{\sigma}_i}|^2
\\
&=W_{\mathcal{E}_o}[\vec{a}]^{-1}
\frac{1}{2^{3n}}\mathbb{E}_{U_i}
\sum_{b_i}\sum_{\vec{a}_i}W_{\mathcal{E}_i}[\vec{a}_i]^{-2}
|\Tr{\rho^TP_{\vec{a}_i}}|^2|\Tr{P_{\vec{a}_i}\hat{\sigma}_i}|^2
\\
&=W_{\mathcal{E}_o}[\vec{a}]^{-1}
\frac{1}{2^{2n}}\sum_{\vec{a}_i}W_{\mathcal{E}_i}[\vec{a}_i]^{-1}
W_{\rho}[\vec{a}].
\end{align*}

\end{proof}

\end{document}